\documentclass[12pt,final,
thmsa,
epsf,a4paper]
{article}
\newtheorem{theorem}{Theorem}

\newtheorem{corollary}{Corollary}
\newtheorem{definition}{Definition}
\newtheorem{example}{Example}
\newtheorem{lemma}{Lemma}
\newtheorem{proposition}{Proposition}
\newtheorem{remark}{Remark}

\newenvironment{proof}[1][Proof]{\emph{#1.} }{\  \hfill $\square $ \vspace{5 pt}}

\usepackage{stackrel}
\usepackage{natbib}
\usepackage{amssymb}
\usepackage[dvips]{graphics}
\usepackage{amsmath}

\usepackage{mathpazo}

\usepackage{enumerate}

\usepackage{amsfonts}

\usepackage{amssymb}

\usepackage{enumerate}

\usepackage{mathrsfs}

\usepackage[usenames]{color}

\usepackage{pgf,tikz}

\usetikzlibrary{decorations.pathreplacing,angles,quotes}
\usetikzlibrary{arrows.meta}
\usetikzlibrary{arrows, decorations.markings}

\usepackage{hyperref}
\hypersetup{
    colorlinks,
    citecolor=blue,
    filecolor=black,
    linkcolor=blue,
    urlcolor=blue
}

\usepackage{pgf,tikz}
\usetikzlibrary{arrows}

\usepackage[utf8]{inputenc}
\usepackage{amssymb, amsmath,geometry}
\usepackage{fancyhdr}

\usepackage{soul}

\usepackage{emptypage}

\newcommand*\samethanks[1][\value{footnote}]{\footnotemark[#1]}

\begin{document}

\title{Lattice structure of the random stable set in many-to-many matching markets \thanks{We would like to thank Agustin Bonifacio, Jordi Mass\'o, Elena Iñarra, Juan Pereyra and the Game Theory Group of IMASL for the helpful discussions and  detailed comments. Our work is partially supported by the UNSL, through grant 319502, and from the Consejo Nacional
de Investigaciones Cient\'{\i}ficas y T\'{e}cnicas (CONICET), through grant
PIP 112-201501-00505, and from Agencia Nacional de Promoción Cient\'ifica y Tecnológica, through grant PICT 2017-2355. }}


\author{Noelia Juárez\thanks{Instituto de Matem\'{a}tica Aplicada San Luis, Universidad Nacional de San
Luis and CONICET, San Luis, Argentina. Emails: \texttt{noemjuarez@gmail.com} (N. Juárez), \texttt{pabloneme08@gmail.com} (P.A. Neme), and \texttt{joviedo@unsl.edu.ar} (J. Oviedo).} \and Pablo A. Neme\samethanks[2] \and Jorge Oviedo\samethanks[2]}

\date{\today}
\maketitle

\begin{abstract}

For a many-to-many matching market, we study the lattice structure of the set of random stable matchings. We define a partial order on the random stable set and present two intuitive binary operations to compute the least upper bound and the greatest lower bound for each side of the matching market. Then, we prove that with these binary operations the set of random stable matchings forms two dual lattices.
\bigskip

\bigskip

\bigskip
\end{abstract}

\section{Introduction}\label{intro}

Matchings have been studied for several decades, beginning with Gale and Shapley's pioneering paper \citep{gale1962college}. They introduce the notion of stable matchings for a marriage market and provide an algorithm for finding them. Since then, a considerable amount of work was carried out on both theory and applications of stable matchings. 
A matching is \textit{stable} if all agents have acceptable partners and there is no pair of agents, one of each side of the market,  that would prefer to be matched to each other rather than to remain with the current partner. Unfortunately, the set of many-to-one stable matchings may
be empty. Substitutability is the weakest condition that has so far been imposed on agents’ preferences under which the existence of stable matchings is guaranteed. An agent has \textit{substitutable preference} if he wants to continue being a partner with agents from the other side of the market even if other agents become unavailable \citep[see][for more detail]{kelso1982job,roth1984evolution}.

One of the most important results in the matching literature is that the set of stable matchings has a dual lattice structure. This is important for at least two reasons. First, it shows that even if agents of the same side of the market compete for agents of the other side, the conflict is attenuated since, on the set of stable matchings, agents on the same side of the market have coincidence of interests. Second, many algorithms are based on this lattice structure. For example, algorithms that yield stable matchings in centralized markets. In this paper, we study the lattice structure of the random stable set for a general matching market, many-to-many matching markets with substitutable preferences and satisfiying the \textit{law of aggregated demand} (L.A.D.). Random stable matchings are very useful for at least two reasons. First, the randomization allows for a much richer space of possible outcomes and may be essential to achieve fairness and anonymity. Second, the framework of random stable matchings admits fractional matchings that capture time-sharing arrangements, \citep[see][among others]{rothblum1992characterization,roth1993stable,teo1998geometry,sethuraman2006many,baiou2000stable,dougan2016efficiency,neme2019characterization,neme2019many}. 

\cite{roth1993stable} define binary operations to compute the  \textit{least upper bound (l.u.b.)} and the \textit{greatest lower bound (g.l.b.)} for random stable matchings for the marriage market. To do so, they use the first-order stochastic dominance as the partial order for random stable matchings. This partial order, can not be applied when agents' preferences are over subsets of agents of the other side of the market in a substitutable manner. For this reason, we present a new partial order --a natural extension of the first-order stochastic dominance-- for the random stable set of a matching market when agents' preferences are subsitutables and satisfies the L.A.D.. Generally, a random stable matching can be represented by different lotteries. Despite this, we prove that there is a unique way to represent a random stable matching fulfilling a special property: each stable matching involve in the lottery of this unique representation is comparable in the eyes of all firms, from now on we refer as \textit{decreasing representation}. This way, our partial order  is independently of the representations of the random stable matching. The process to construct this decreasing representation for each random stable matchings, its presented as Algorithm 1.

To ease the definition of the binary operations and proofs, we present the \textit{splitting procedure}. Given two random stable matchings, this procedure ``splits'' the decreasing representation of each random stable matching in a way that both lotteries have the same numbers of terms. Moreover, both lotteries have the same scalars, term to term.  The splitting procedure is formalized by Algorithm 2 presented in Appendix \ref{apendice B}. 

Our main contribution in this paper is that, by defining two natural binary operations (pointing functions) that compute the \textit{l.u.b.} and \textit{g.l.b.} for random stable matchings, the set of these matchings has a dual lattice structure. In other words, as long as the set of (deterministic) stable matchings  has a lattice structure, where the binary operations are computed via pointing functions, the set of random stable matchings also has a lattice structure.

Further, for the special case in which all the scalars of the lottery are rational numbers, we show that there is a direct way to compute the \textit{l.u.b.} and the \textit{g.l.b.}

Other finding derived from our proofs, is a version of Rural Hospital Theorem for random stable matchings. The Rural Hospital Theorem (for deterministic stable matchings) for a many-to-many matching market where all agents have substitutable preference satisfying the L.A.D. is presented in \cite{alkan2002class}. 
The paper illustrates the successive results needed to prove that the random sable set has a lattice structure with numeric examples.

\subsection*{Related literature}
The lattice structure of the set of stable matchings is introduced by \cite{knuth1976marriages} for the marriage market. Given two stable matchings he defines the \textit{l.u.b.} for men, by matching to each man with the best of the two partners, and the \textit{g.l.b.} for men, by matching to each man the less preferred between the two partners; these are usually called the \textit{pointing functions} relative to a partial order.  \cite{roth1985college} shows that these binary operations (pointing functions) used in \cite{knuth1976marriages} do not work in the more general many-to-many and many-to-one matching markets introduced by \cite{kelso1982job} and \cite{roth1984evolution} respectively even under substitutable preferences. For a specific many-to-one matching market, the so-called the college admission problem, \cite{roth1992two} present a natural extension of Knuth's result for $q$-responsive preferences. \cite{marti2001lattice} further extend the results proved by \cite{roth1992two}. They identified a weaker condition than $q$-responsiveness, called $q$-separability, and propose two natural binary operations that give a dual lattice structure to the set of stable matchings in a many-to-one matching market with substitutable and $q$-separable preferences. Such binary operations are similar to the Knuth’s ones. \cite{risma2015binary} generalizes the result of \cite{marti2001lattice} by showing that their binary operations work well in many-to-one matching markets where the preferences of the agents satisfy substitutability and the \textit{law of aggregate demand} (a less restrictive than $q$-separability). Her paper is contextualized in many-to-one matching markets with contracts. 
 \cite{manasero2019binary} extends the result in \cite{risma2015binary} to the many-to-many marching market, where one side has substitutable preferences satisfying the law of aggregate demand, and the other side has $q$-responsive preferences. \cite{alkan2002class}  considers a market with multiple partners on both sides. For this market, preferences are given by rather general \textit{path-independent choice functions} that do not necessarily respect any ordering on individuals and satisfy the law of aggregated demand.\footnote{\cite{alkan2002class} calls ``the law of aggregated demand'' as ``cardinal monotonicity''.} He shows that the set of stable matchings in any two-sided market with path-independent  choice functions and preferences satisfying the law of aggregated demand has a lattice structure under the common preferences of all agents on any side of the market.  \cite{li2014new}  presents an alternative proof for Alkan's result. The main distinction between \cite{li2014new} and \cite{alkan2002class} lies in the conditions over preferences: \cite{li2014new} assumes agents with complete preferences, whereas \cite{alkan2002class} assumes agents with incomplete revealed preferences. All of these mentioned papers share natural definitions of the binary operations via pointing functions. 
 
In other direction, there is an extensive literature that proves that the set of stable matchings has a lattice structure \citep[see][among others]{blair1988lattice,adachi2000characterization,fleiner2003fixed,echenique2004core,echenique2004theory,hatfield2005matching,ostrovsky2008stability,wu2018lattice}. All of these mentioned papers define in a difficult way, by means of fixed points, the \textit{l.u.b} and \textit{g.l.b.}. That is, these papers do not compute the binary operations via pointing functions.

Regarding to the related literature concerning to lattice structures of random stable sets, \cite{roth1993stable}, define two binary operations for random stable matchings in marriage markets. For these very particular markets, they proved that the set of random stable matchings,\footnote{They prove that the ``stable fractional matching set'' coincides with random stable matching set} endowed with a partial order (first-order stochastic dominance) has a lattice structure. \cite{neme2019characterization} proved that the strongly stable fractional matching set in the marriage market, endowed with the same partial order (first-order stochastic dominance), has a lattice structure. The binary operations defined in \cite{roth1993stable} and also used by \cite{neme2019characterization}, can not be extended to a more general markets, not even to the college admission problem with $q$-responsive preferences. 

The paper is organized as follows. In Section \ref{prelimirany}, we introduce the matching market and preliminary results. In Section \ref{seccion algoritmo}, we prove that there is a unique way to represent a random stable matching with a decreasing property (Algorithm 1 and Theorem \ref{teorema del orden}). Also, we present a version of the Rural Hospital Theorem for random stable matchings (Proposition \ref{hospital rural para random}). In section \ref{order}, we present a partial order for random matchings when agents' preferences are substitutalbes and satisfies L.A.D. (Proposition \ref{Es orden parcial}). Also, we describe the splitting procedure that is formalized latter in Appendix \ref{apendice B}. In Section \ref{seccion main result}, we define the binary operations, and we prove that these natural binary operations computes the \textit{l.u.b.} and \textit{g.l.b.} for each side of the market (Proposition \ref{teorema de operaciones binarias}). Further,  the main result of the paper is presented, and  states that the random stable set has a dual lattice structure (Theorem \ref{teorema principal}). In subsection \ref{rational}, we show how to compute in a direct way the \textit{l.u.b.} and \textit{g.l.b.} for rational random stable matchings (these are random stable matchings where all scalars of their lotteries are rational numbers)(Corollary \ref{corolario para rational}). Section \ref{conclusiones} contains concluding remarks. Finally, Appendix \ref{apendice A} contains proofs for the decreasing representation and Appendix \ref{apendice B} contains proofs of the partial order, formalization of the splitting procedure (Algorithm 2), and the proof of the main theorem.

\section{Preliminaries}\label{prelimirany}
We consider many-to-many matching markets, where there are two disjoint sets of agents, the set of \textit{firms} $F$ and the set of \textit{workers} $W$. Each firm has an antisymmetric, transitive and complete preference relation ($>_f$) over the set of all subsets of $W$.  In the same way, each worker has an antisymmetric, transitive and complete preference relation ($>_w$) over the set of all subsets of $F$. We denote by $P$ the preferences profile for all agents: firms and workers. A many-to-many matching market is denoted by $(F,W,P).$ 
Given a set of firms $S\subseteq F$, each worker $w\in W$ can determine which subset of $S$ would most prefer to hire. We call this the $w$'s choice set from $S$ and denote it by $Ch(S,>_w)$. Formally,
$$
Ch(S,>_w)=max_{>_w}\{T:T\subseteq S\}.
$$
Symmetrically, given a set of workers $S\subseteq W$, let $Ch(S,>_f)$ denote firm $f$'s most preferred subset of $S$ according to its preference relation $>_f$. Formally,
$$
Ch(S,>_f)=max_{>_f}\{T:T\subseteq S\}.
$$

\begin{definition}
A \textbf{matching} $\mu$ is a function from the set $F\cup W$ into $2^{F\cup W}$ such that for each $w\in W$ and for each $f\in F$:
\begin{enumerate}[1.]
\item $\mu(w)\subseteq F$;
\item $\mu(f)\subseteq W$;
\item $w\in \mu(f)\Leftrightarrow f\in \mu(w)$
\end{enumerate}
\end{definition}
We say that agent $a\in F\cup W$ is matched if $\mu(a) \neq \emptyset$, otherwise he is unmatched.

A matching $\mu$ is blocked by agent a if $\mu(a)\neq Ch(\mu(a),>_a)$. We say that a matching is individually rational if it is not blocked by any individual agent. A matching $\mu$ is blocked by a worker-firm pair $(w,f)$ if $w \notin \mu( f ), w \in Ch(\mu( f )\cup \{w\},>_f),$
and $f \in Ch(\mu( w )\cup \{f\},>_w)$. A matching $\mu$ is \textbf{stable} if it is not blocked by any individual agent or any worker-firm pair. The set of stable matchings is denoted by $\boldsymbol{\mathcal{S(P)}}.$ Further, a \textbf{random stable matching} is a lottery over stable matchings, and denote by $\boldsymbol{\mathcal{RS(P)}}$ the random stable set for the many-to-many matching market $(F,W,P)$. 

Given an agent $a$'s preference relation ($>_a$) and two stable matchings $\mu$ and  $\mu'$, let $\mu(a) \geq_a \mu'(a)$ denote $\mu(a)=Ch(\mu(a) \cup \mu'(a),>_a)$. We say that $\mu(a) >_a \mu'(a)$ if $\mu(a) \geq_a \mu'(a)$ and $\mu(a) \neq \mu'(a)$. Given a preferences profile $P$, and two stable matchings $\mu$ and  $\mu'$, let  $\mu >_F \mu'$ denote the case in which all firms like $\mu$ at least as well as $\mu'$, with at least one firm preferring $\mu$ to $\mu'$ outright. Let $\mu \geq_F \mu'$ denote that either $\mu >_F \mu'$ or that $\mu= \mu'.$ Similarly, we define $>_W$ and $\geq_W$. Notice that, $\geq_F$ and  $\geq_W$ are partial orders over the set of stable matchings.

An agent $a$'s preferences relation satisfies \textbf{substitutability} if, for any subset $S$ of the opposite set (for instance, if $a\in F$ then $S\subseteq W$) that contains agent $b$, $b\in Ch(S,>_a)$ implies $b\in Ch(S'\cup \{b\},>_a)$ for all $S' \subseteq S.$ We say that an agent $a$'s preference relation ($>_a$) satisfies the\textbf{ law of aggregated demand (L.A.D.)} if for all subset $S$ of the opposite set and all $S'\subseteq S$, $|Ch(S',>_a)|\leq |Ch(S,>_a)|.$ \footnote{$|S|$ denotes the number of agents in $S$.}

For a matching market $(F,W,P)$ where the preference relation of each agent satisfies substitutability and the LAD, \cite{alkan2002class}\footnote{\cite{li2014new} present an alternative proof for Alkan's result, \cite{li2014new} assumes agents with complete preferences, whereas \cite{alkan2002class} assumes agents with incomplete preferences.} proves that the set of stable matchings has a lattice structure.  Given two stable matchings $\mu_1$ and $\mu_2$, \textit{l.u.b.} for firms is denoted by $\boldsymbol{\mu_1 \vee_F \mu_2}$ and  \textit{g.l.b.} for firms is denoted by $\boldsymbol{\mu_1 \wedge_F \mu_2}$. Similarly, \textit{l.u.b.} for workers is denoted by $\boldsymbol{\mu_1 \vee_W \mu_2}$ and \textit{g.l.b.} for workers is denoted by $\boldsymbol{\mu_1 \wedge_W \mu_2}$. The binary operations are defined as follows, \citep[see][among others]{alkan2002class,li2014new}.
$$\mu_1 \vee_F \mu_2(f)=\mu_1 \wedge_W \mu_2(f):=Ch(\mu_1(f)\cup \mu_2(f),>_f),\text{ for each firm }f\in F,$$
$$\mu_1 \vee_F \mu_2(w)=\mu_1 \wedge_W \mu_2(w):=\{f:w\in Ch(\mu_1(f)\cup \mu_2(f),>_f)\},\text{ for each worker }w\in W.$$
Similarly, 
$$\mu_1 \vee_W \mu_2(w)=\mu_1 \wedge_F \mu_2(w):=Ch(\mu_1(w)\cup \mu_2(w),>_w),\text{ for each worker }w\in W, $$
$$\mu_1 \vee_W \mu_2(f)=\mu_1 \wedge_F \mu_2(f):=\{w:f\in Ch(\mu_1(w)\cup \mu_2(w),>_w)\},\text{ for each firm }f\in F.$$
\begin{remark}\label{operaciones en matching es estable}
Let $T\subseteq \mathcal{S(P)}$. We denote by $$\bigvee_{\nu\in{T}}\displaystyle{_{\!\!\!\!{F}}}~~\nu(f)=Ch(\bigcup_{\nu\in{T}}\nu(f),>_f)$$ and  $$\bigwedge_{\nu\in{T}}\displaystyle{_{\!\!\!{F}}}~~\nu(f)=\{w:f\in Ch(\bigcup_{\nu\in{T}}\nu(w),>_w)\}.$$ By substitutability and transitivity, \cite{li2014new} proves that  that  $$\bigvee_{\nu\in{T}}\displaystyle{_{\!\!\!\!{F}}}~~\nu(f)\text{ and }\bigwedge_{\nu\in{T}}\displaystyle{_{\!\!\!{F}}}~~\nu(f)$$ are stable matchings and coincide with the \textit{l.u.b.} and \textit{g.l.b.} among the stable matchings in $T$ respectively.
\end{remark}

\section{Random stable matchings: representations}\label{seccion algoritmo}
In this section, we present two results that have interest in themselves and that we use in the next section in order to prove that the set of random stable matchings has a lattice structure. Given a random stable matching that is represented as a lottery over stable matchings, we change its representation as a new lottery over a new set of stable matching. To be more specific, we prove that this new set of stable matchings have a decreasing property, namely, there is $\{\mu_1,\ldots,\mu_{\tilde{k}} \}\subseteq \mathcal{S(P)}$ with $\mu_\ell \geq_F \mu_{\ell+1}$ for $\ell=1,\ldots,{\tilde{k}}-1$. Also, we present a version for random stable matchings of the Rural Hospital Theorem (Proposition (RHT)).


To describe the representation of a random stable matchings, first, we need to define an incidence vector. Then, given a stable matching $\mu$, a vector $x^{\mu}\in\left\{  0,1\right\}^{\left\vert F\right\vert \times\left\vert W\right\vert }$ is an \textbf{incidence vector} where $x_{i,j}^{\mu}=1$ if and only if $j\in{\mu\left(  i\right)}$ and $x_{i,j}^{\mu}=0$ otherwise.
Hence, a random stable matching is represented as a lottery over the incidence vectors of stable matchings. That is,
$$
x=\sum_{\nu\in{\mathcal{S(P)}}} \lambda_{\nu} x^{\nu}
$$
where $0\leq \lambda_{\nu} \leq 1,~\sum_{\nu\in{\mathcal{S(P)}}}\lambda_{\nu}=1, \mbox{ and } \nu\in{\mathcal{S(P)}}$. 

Notice that, each entry of a random stable matching $x$, fulfils that $x_{i,j} \in [0,1]$.

Given a random stable matching $x$, we define the support of $x$ as follows:
$$
supp(x)=\{(i,j):x_{i,j}>0\}.
$$

Given a random stable matching $x$, i.e. $x=\sum_{\nu\in{\mathcal{S(P)}}}\lambda_{\nu} x^{\nu};~0 \leq \lambda_{\nu} \leq 1,~\sum_{\nu\in{\mathcal{S(P)}}}\lambda_{\nu}=1,$ we define $A$ to be the set of all stable matchings involve in the lottery. Formally, 
$$
A=\bigg\{ \nu\in{\mathcal{S(P)}}: x=\sum_{\nu\in{\mathcal{S(P)}}}\lambda_{\nu} x^{\nu};~0< \lambda_{\nu} \leq 1,~\sum_{\nu\in{\mathcal{S(P)}}}\lambda_{\nu}=1 \bigg\}.
$$
Now, in order to change the representation of the random stable matching $x$ proceed as follows:

\begin{center}
\begin{tabular}{l l}
\hline \hline
\multicolumn{2}{l}{\textbf{Algorithm 1:}}\vspace*{10 pt}\\
\textbf{Step $\boldsymbol{0}$} &  Set $B_1:=A ~~\displaystyle{\bigcup} ~\bigg\{\displaystyle{\bigvee_{\nu\in{T}}}\displaystyle{_{\!\!{F}}}~~\nu:T\subseteq A\bigg\}\bigcup\bigg\{\bigwedge_{\nu\in{T}}\displaystyle{_{\!\!{F}}}~~\nu:T\subseteq A\bigg\} .$\\
& \hspace{20 pt}$x^1:=x$.\\
& \hspace{20 pt}$\mathcal{M}:=\emptyset$.\\
& \hspace{20 pt}$\Lambda:=\emptyset$.\\
\textbf{Step $\boldsymbol{k\geq1}$} &  Set $\mu_{k}:=\displaystyle{\bigvee_{\nu\in{B_{k}}}}\displaystyle{_{\!\!\!{F}}}~~\nu.$\\
& \hspace{20 pt}$\mathcal{M}:=\mathcal{M}\cup \{\mu_k\}.$\\
& \hspace{20 pt}$\alpha_{k}:=min\{x^k_{i,j} :x^{\mu_{k}}_{i,j}=1 \}.$\\
& \hspace{20 pt}$\Lambda:=\Lambda \cup \{\alpha_k\}.$\\
&\hspace{20 pt}$\mathcal{L}_k:=\{(i,j)\in F\times W: x^k_{i,j}=\alpha_k \text{ and }x_{i,j}^{\mu_k}=1\}.$\\
& \hspace{20 pt}$C_k:=\displaystyle\bigcup_{(i,j)\in{\mathcal{L}_k}}\{\nu \in B_k:x_{i,j}^{\nu}=1\}.$\\
& \hspace{20 pt}$B_{k+1}:=B_k \setminus C_k.$ \\
& \texttt{IF} $B_{k+1}=\emptyset$,\\
& \hspace{20 pt}\texttt{THEN}, the procedure stops. \\
& \texttt{ELSE} set $x^{k+1}=\frac{x^k-\alpha_kx^{\mu_k}}{1-\alpha_k},$ and  continue to Step $k+1.$\medskip\\
\hline \hline
\end{tabular}
\end{center}

The following theorem states that there is a unique representation of a random stable matching with the decreasing property in the eyes of all firms.

\begin{theorem}\label{teorema del orden}
Let $x$ be a random stable matching and $\mathcal{M}$ be the output of Algorithm 1. Then, $x$ is represented as a lottery over stable matchings that belong to $\mathcal{M}$ where $\mu_{\ell}>_{F} \mu_{\ell+1}$ for each $\mu_\ell,\mu_{\ell+1} \in \mathcal{M}$. Moreover, the set $\mathcal{M}$ is unique.
\end{theorem}

\begin{proof}
See proof in Appendix \ref{apendice A}.
\end{proof}

The following example illustrate Algorithm 1.

\begin{example}\label{ejempplo1}
Let $(F,W,P)$ be a many-to-one matching market instance where $F=\{f_1,f_2,f_3,f_4\}$, $W=\{w_1,w_2,w_3,w_4\}$ and the preference profile is given by

$$
\begin{array}{c}
>_{f_1}=\{w_1,w_2\},\{w_1,w_3\},\{w_2,w_4\},\{w_3,w_4\},\{w_1\},\{w_2\},\{w_3\},\{w_4\}.  \\
>_{f_2}=\{w_3,w_4\},\{w_2,w_4\},\{w_1,w_3\},\{w_1,w_2\},\{w_3\},\{w_4\},\{w_1\},\{w_2\}.  \\
>_{f_3}=\{w_1,w_3\},\{w_3,w_4\},\{w_1,w_2\},\{w_2,w_4\},\{w_1\},\{w_3\},\{w_2\},\{w_4\}.  \\
>_{f_4}=\{w_2,w_4\},\{w_1,w_2\},\{w_3,w_4\},\{w_1,w_3\},\{w_2\},\{w_4\},\{w_1\},\{w_3\}.  \\
>_{w_1}=\{f_2,f_4\},\{f_2,f_3\},\{f_1,f_4\},\{f_1,f_3\},\{f_2\},\{f_4\},\{f_3\},\{f_1\}.  \\
>_{w_2}=\{f_2,f_3\},\{f_1,f_3\},\{f_2,f_4\},\{f_1,f_4\},\{f_2\},\{f_3\},\{f_1\},\{f_4\}.  \\
>_{w_3}=\{f_1,f_4\},\{f_2,f_4\},\{f_1,f_3\},\{f_2,f_3\},\{f_1\},\{f_4\},\{f_2\},\{f_3\}.  \\
>_{w_4}=\{f_1,f_3\},\{f_1,f_4\},\{f_2,f_3\},\{f_2,f_4\},\{f_1\},\{f_3\},\{f_4\},\{f_2\}.  \\
\end{array}
$$
It is easy to check that these preference relations are substitutable and satisfy LAD.
The set of stable matchings is represented in Table 1 and its lattice for the partial order $\geq_F$ is represented in Figure 1.

\bigskip
\begin{minipage}{0.7\linewidth}
\begin{center}
\begin{tabular}{c|cccc}
& $\boldsymbol{f_1}$ & $\boldsymbol{f_2}$ & $\boldsymbol{f_3}$ & $\boldsymbol{f_4}$  \\ \hline 
$\boldsymbol{\nu_1}$ & $\{w_1,w_2\}$ & $\{w_3,w_4\}$ & $\{w_1,w_3\}$ & $\{w_2,w_4\}$ \\
$\boldsymbol{\nu_2}$ & $\{w_1,w_3\}$ & $\{w_2,w_4\}$ & $\{w_3,w_4\}$ & $\{w_1,w_2\}$ \\
$\boldsymbol{\nu_3}$ & $\{w_2,w_4\}$ & $\{w_1,w_3\}$ & $\{w_1,w_2\}$ & $\{w_3,w_4\}$ \\
$\boldsymbol{\nu_4}$ & $\{w_3,w_4\}$ & $\{w_1,w_2\}$ & $\{w_2,w_4\}$ & $\{w_1,w_3\}$ \\
\end{tabular}

\medskip

\hspace {20pt} Table 1

\end{center}
\end{minipage}
\begin{minipage}{0.3\linewidth}
{\small
\hspace{15pt}\begin{tikzpicture}[scale=0.35]
\node (1) at (2,18) {\small{$\nu_1$}};
\node (2) at (0,15) {\small{$\nu_2$}};
\node (3) at (4,15) {\small{$\nu_3$}};
\node (4) at (2,12) {\small{$\nu_4$}};

\draw(1) to (2);
\draw(1) to (3);
\draw(1) to (2);
\draw(2) to (4);
\draw(3) to (4);

\end{tikzpicture}}

\hspace {20pt} Figure 1
\end{minipage}

\bigskip

Let $x^1=\frac{3}{4}x^{\nu_2}+\frac{1}{4}x^{\nu_3}$ be a random stable matching. Now, we change the representation of $x^1$ as in Theorem \ref{teorema del orden}.
Notice that 
$$
x^1=
\left(\begin{matrix}
\frac{3}{4} &\frac{1}{4} & \frac{3}{4} & \frac{1}{4} \\
\frac{1}{4} &\frac{3}{4} & \frac{1}{4} & \frac{3}{4}  \\
\frac{1}{4} &\frac{1}{4} & \frac{3}{4} & \frac{3}{4} \\
\frac{3}{4} & \frac{3}{4} &\frac{1}{4} & \frac{1}{4}  \\
\end{matrix}\right).
$$

Then, $A=\{\nu_2, \nu_3\}$, and $B_1=\{\nu_1,\nu_2,\nu_3,\nu_4\}.$
Set $\mathcal{M}:=\emptyset$ and $\Lambda:=\emptyset.$

\noindent \textbf{Step 1}
Set $\mu_1 := \nu_1=\displaystyle \bigvee_{\nu\in B_1}\displaystyle{_{\!\!\!\!{F}}}~\nu, ~\mathcal{M}:=\mathcal{M}\cup \{\mu_1\}, ~\alpha_1=\frac{1}{4}, ~\Lambda:=\Lambda\cup \{\alpha_1\} $ and $C_1=\{\nu_1,\nu_3\}.$ 

Since $B_2=B_1 \setminus C_1=\{\nu_2,\nu_4\}\neq \emptyset, $ then set
$$
x^2:=\frac{x-\frac{1}{4}x^{\mu_1}}{1-\frac{1}{4}}=\left(\begin{matrix}
\frac{2}{3} & 0 &1 & \frac{1}{3}  \\
\frac{1}{3} &1 & 0 & \frac{2}{3}  \\
0 & \frac{1}{3}  &\frac{2}{3} & 1 \\
 1 & \frac{2}{3} &\frac{1}{3} & 0  \\
\end{matrix}\right),
$$
and continue to Step 2.

\noindent \textbf{Step 2}
Set $\mu_2 := \nu_2=\displaystyle \bigvee_{\nu\in B_2}\displaystyle{_{\!\!\!\!{F}}}~\nu, ~\mathcal{M}:=\mathcal{M}\cup \{\mu_2\}, ~\alpha_2=\frac{2}{3}, ~\Lambda:=\Lambda\cup \{\alpha_2\} $ and $C_2=\{\nu_2\}.$ 

Since $B_3=B_2 \setminus C_2=\{\nu_4\}\neq \emptyset, $ then set
$$
x^3=\frac{x-\frac{2}{3}x^{\mu_2}}{1-\frac{2}{3}}=\left(\begin{matrix}
0 & 0 & 1 & 1 \\
1 &1 & 0 & 0  \\
0 & 1 &0 & 1 \\
 1 & 0 &1 & 0  \\
\end{matrix}\right),
$$
and continue to Step 3.

\noindent \textbf{Step 3}
Set $\mu_3 := \nu_4=\displaystyle \bigvee_{\nu\in B_3}\displaystyle{_{\!\!\!\!{F}}}~\nu, ~\mathcal{M}:=\mathcal{M}\cup \{\mu_3\}, ~\alpha_3=1, ~\Lambda:=\Lambda\cup \{\alpha_3\} $ and $C_3=\{\nu_4\}.$ 

Since $B_4=B_3\setminus C_3= \emptyset, $ then the procedure stops.

The output of Algorithm 1 is $\mathcal{M}=\{\mu_1,\mu_2,\mu_3\}=\{\nu_1,\nu_2,\nu_4\},$ and $\Lambda=\{\frac{1}{4},\frac{2}{3},1\}.$

Therefore,
\begin{center}
$x^1=\frac{1}{4}x^{\mu_1}+(1-\frac{1}{4})(\frac{2}{3})x^{\mu_2}+(1-\frac{1}{4})(1-\frac{2}{3})(1)x^{\mu_3}$
\end{center} 
\begin{center}
$
=\frac{1}{4}x^{\mu_1}+\frac{1}{2}x^{\mu_2}+\frac{1}{4}x^{\mu_3}.
$
\end{center}

Since $\mu_1=\nu_1$, $\mu_2=\nu_2$ and $\mu_3=\nu_4$, then $x^1$ can be written as:

\begin{center}
$x^1=\frac{1}{4}x^{\nu_1}+\frac{1}{2}x^{\nu_2}+\frac{1}{4}x^{\nu_4}.
$
\end{center}

As we can see in Figure 1, the stable matchings of the lottery fulfils $\nu_1 \geq_F \nu_2 \geq_F  \nu_4$.
\end{example}

The following proposition it is known as \textit{Rural Hospital Theorem}. For a many-to-many matching markets where the preference relation of each agent satisfies substitutability and the LAD is proved in \cite{alkan2002class}.
\medskip

\noindent \textbf{Proposition (RHT) (\cite{alkan2002class})} \textit{
Each agent is matched with the same number of partners in every stable matching. That is, $|\mu(a)|=|\mu'(a)|$ for each $\mu,\mu'\in \mathcal{S(P)}$ and for each $a\in F\cup W$.}
\medskip

Next, we present a version for random stable matchings of Proposition (RHT).
\begin{proposition}\label{hospital rural para random}
Let $x$ and $x'$ be two random stable matchings, then $\sum_{i\in F}x_{i,j}=\sum_{i\in F}x'_{i,j}$ for each $j\in W$, and $\sum_{j\in W}x_{i,j}=\sum_{j\in W}x'_{i,j}$ for each $i\in F$.
\end{proposition}

\begin{proof}
See proof in Appendix \ref{apendice A}. 
\end{proof}

From now on, by Proposition \ref{teorema del orden}, we assume that each random stable matching is already represented as a lottery over stable matchings in a decreasing way.

\section{Partial order for random stable matchings}\label{order}

In this section, we define a  partial order for the set of random stable set in a many-to-many matching market with substitutable preferences satisfying the L.A.D.. This partial order is a generalization of the first-order stochastic dominance presented in \cite{roth1993stable} for the random stable set in the marriage market.  
Given two random stable matching $x$ and $y$ for the marriage market $(M,W,P)$, \cite{roth1993stable} define the partial order as follows.
They say that $x$ weakly dominates$^\star$ $y$ for man $m$,  (here denoted by $x \succeq_m^\star y$) if

$$
\sum_{j\geq_mw} x_{m,j} \geq \sum_{j\geq_mw} y_{m,j}
$$ for each $w\in W$. Further they say that $x \succeq_M^\star y$ if $x \succeq_m^\star y$ for each $m\in M$. The partial order $\succeq^\star_W$ is defined analogously.
Notice that the partial orders $\succeq_M^\star$ and $\succeq_W^\star$ can not order random stable matchings when agents have preferences over subset of agents on the other side of the market in a substitutable manner. For this reason, for the setting considered in this paper, we define a new partial order. Formally,
\begin{definition}\label{defino orden estocastico}
Let $x=\sum_{i=1}^{I}\alpha_{i}x^{\mu^{x}_{i}}$ and $y=\sum_{j=1}^{J}\beta_{j}x^{\mu^{y}_{j}}$ with $\mu^{x}_i \geq_F \mu^{x}_{i +1}$ for each $i=1,\ldots,I-1$ and $\mu^{y}_j \geq_F \mu^{y}_{j +1}$ for each $j=1,\ldots,J-1$. 
We say that $\boldsymbol{x}$ \textbf{weakly dominates} $\boldsymbol{y}$ for the firm $f$, ($x\succeq_{f} y$), if and only if
for each $\mu^y_{j}(f)$
$$
\sum_{i:\mu_{i}^{x}(f)\geq_{f}\mu_{j}^{y}(f)}\alpha_{i} \geq \sum_{l:\mu_{l}^{y}(f)\geq_{f}\mu_{j}^{y}(f)}\beta_{l}.
$$
\end{definition}
Further, we say that $\boldsymbol{x}$ \textbf{strongly dominates} $\boldsymbol{y}$ for the firm $f$, ($x\succ_f y$), if the above inequalities hold with at least one strict inequality for some $\mu_j^y(f)$. That is, $x\succ_f y$ when $x\succeq_f y$ and $x\neq y$ for the firm $f$. Further, if $x\succeq_{f} y$ for each $f\in F$ we denote that $x\succeq_{F} y$. We define $\succeq_{w}$, $\succ_{w}$ and $\succeq_{W}$ analogously. If we interpret the $x_{f,w}$ as the probability that $f$ is matched with $w$, then $x\succeq_f y$ if $x_{f,\cdot}$ stochastically dominates $y_{f,\cdot}$. 
Notice that, since both $x$ and $y$ are represented following Theorem \ref{teorema del orden}, then  $$\sum_{l:\mu_{l}^{y}(f)\geq_{f}\mu_{j}^{y}(f)}\beta_{l} =\sum_{l=1}^{j}\beta_{l}.$$

Now we prove that the domination relation $\succeq_{F}$ is a partial order. The proof of $\succeq_{W}$ is analogously. Formally,

\begin{proposition}\label{Es orden parcial}
The domination relation $\succeq_{F}$ is a partial order.
\end{proposition}
\begin{proof}
See proof in Appendix \ref{apendice B}.
\end{proof}

\subsection{Splitting procedure}

In this subsection we explain the splitting procedure for two random stable matchings, that is formalized with an algorithm in Appendix \ref{apendice B}. Once we apply the splitting procedure for two random stable matchings, we define a domination relation ($\succeq_F^S$) that is further used to define the binary operation in a simple way.
Given two random stable matchings  $x=\sum_{i=1}^{I}\alpha_{i}x^{\mu^{x}_{i}}$ and $y=\sum_{j=1}^{J}\beta_{j}x^{\mu^{y}_{j}}$ with $\mu^{x}_i \geq_F \mu^{x}_{i +1}$ for each $i=1,\ldots,I-1$ and $\mu^{y}_j \geq_F \mu^{y}_{j +1}$ for each $j=1,\ldots,J-1$, the splitting procedure goes as follows:
Let $\gamma_1=min\{\alpha_1, \beta_1\}$. W.l.o.g. assume that $\gamma_1=\alpha_1.$ Then,
$$x=\gamma_1 \mu^x_1+\sum_{i=2}^{I}\alpha_{i}x^{\mu^{x}_{i}}$$
$$y=\gamma_1  \mu_1^y+(\beta_1-\gamma_1)\mu^{y}_1+\sum_{j=2}^{J}\beta_{j}x^{\mu^{y}_{j}}.$$
Notice that the first terms of each new representation have the same scalar. Now, take the second scalar of each representation and set $\gamma_2=min\{\alpha_2, \beta_1-\gamma_1\}$. If $\gamma_2=\alpha_2$, then 
$$x=\gamma_1 \mu^x_1+\gamma_2 \mu^x_2+\sum_{i=3}^{I}\alpha_{i}x^{\mu^{x}_{i}}$$
$$y=\gamma_1  \mu_1^y+\gamma_2 \mu^y_1+(\beta_1-\gamma_1-\gamma_2)\mu^{y}_1+\sum_{j=2}^{J}\beta_{j}x^{\mu^{y}_{j}}.$$
If $\gamma_2=\beta_1-\gamma_1$, then 
$$x=\gamma_1 \mu^x_1+\gamma_2 \mu^x_2+(\alpha_2-\gamma_2) \mu^x_2+\sum_{i=3}^{I}\alpha_{i}x^{\mu^{x}_{i}}$$
$$y=\gamma_1  \mu_1^y+\gamma_2 \mu^y_1+\sum_{j=2}^{J}\beta_{j}x^{\mu^{y}_{j}}.$$
Notice that the first two terms of each new representation have the same scalar. Now take the third scalar of each representation and set either $\gamma_3= min\{\alpha_3, \beta_1-\gamma_1-\gamma_2\}$ or $\gamma_3=min \{\alpha_2-\gamma_2, \beta_2\}$, and so forth so on.

We illustrate the splitting procedure with the following example.
 
\noindent \textbf{Example 1 (Continued)} \textit{Let $x=\frac{1}{4}x^{\nu_1}+\frac{1}{2}x^{\nu_2}+\frac{1}{4}x^{\nu_4}$ and $y=\frac{1}{6}x^{\nu^1}+\frac{1}{2}x^{\nu^3}+\frac{1}{3}x^{\nu^4}$. Notice that both random stable matchings are represented following Theorem \ref{teorema del orden}. 
Let $\gamma_1=min\{\frac{1}{4},\frac{1}{6}\}=\frac{1}{6},$ then
\begin{center}
$x=\frac{1}{6}x^{\nu_1}+(\frac{1}{4}-\frac{1}{6})x^{\nu_1}+\frac{1}{2}x^{\nu_2}+\frac{1}{4}x^{\nu_4},$
\end{center}
\begin{center}
$y=\frac{1}{6}x^{\nu^1}+\frac{1}{2}x^{\nu^3}+\frac{1}{3}x^{\nu^4}$
\end{center}
Notice that the first term of each new representation have the same scalar $\frac{1}{6}$.
Let $\gamma_2=min\{\frac{1}{4}-\frac{1}{6},\frac{1}{2}\}=\frac{1}{4}-\frac{1}{6}=\frac{1}{12},$ then
\begin{center}
$x=\frac{1}{6}x^{\nu_1}+\frac{1}{12}x^{\nu_1}+\frac{1}{2}x^{\nu_2}+\frac{1}{4}x^{\nu_4},$
\end{center}
\begin{center}
$y=\frac{1}{6}x^{\nu^1}+\frac{1}{12}x^{\nu_3}+(\frac{1}{2}-\frac{1}{12})x^{\nu^3}+\frac{1}{3}x^{\nu^4}.$
\end{center} 
Notice that the second term of each new representation also have the same scalar $\frac{1}{12}$.
Let $\gamma_3=min\{\frac{1}{2},\frac{1}{2}-\frac{1}{12}\}=\frac{1}{2}-\frac{1}{12}=\frac{5}{12},$ then
\begin{center}
$x=\frac{1}{6}x^{\nu_1}+\frac{1}{12}x^{\nu_1}+\frac{5}{12}x^{\nu^2}+(\frac{1}{2}-\frac{5}{12})x^{\nu_2}+\frac{1}{4}x^{\nu_4},$
\end{center}
\begin{center}
$y=\frac{1}{6}x^{\nu^1}+\frac{1}{12}x^{\nu_3}+\frac{5}{12}x^{\nu^3}+\frac{1}{3}x^{\nu^4}.$
\end{center} 
Notice that the third term of each new representation also have the same scalar $\frac{5}{12}$.
Let $\gamma_4=min\{\frac{1}{2}-\frac{5}{12},\frac{1}{3}\}=\frac{1}{2}-\frac{1}{12}=\frac{1}{12},$ then
\begin{center}
$x=\frac{1}{6}x^{\nu_1}+\frac{1}{12}x^{\nu_1}+\frac{5}{12}x^{\nu^2}+\frac{1}{12}x^{\nu_2}+\frac{1}{4}x^{\nu_4},$
\end{center}
\begin{center}
$y=\frac{1}{6}x^{\nu^1}+\frac{1}{12}x^{\nu_3}+\frac{5}{12}x^{\nu^3}+\frac{1}{12}x^{\nu_4}+(\frac{1}{3}-\frac{1}{12})x^{\nu^4}.$
\end{center}
Notice that the fourth term of each new representation also have the same scalar $\frac{1}{12}$.
Let $\gamma_4=min\{\frac{1}{4},\frac{1}{3}-\frac{1}{12}\}=min\{\frac{1}{4},\frac{1}{4}\}=\frac{1}{4},$ then
\begin{center}
$x=\frac{1}{6}x^{\nu_1}+\frac{1}{12}x^{\nu_1}+\frac{5}{12}x^{\nu^2}+\frac{1}{12}x^{\nu_2}+\frac{1}{4}x^{\nu_4},$
\end{center}
\begin{center}
$y=\frac{1}{6}x^{\nu^1}+\frac{1}{12}x^{\nu_3}+\frac{5}{12}x^{\nu^3}+\frac{1}{12}x^{\nu_4}+\frac{1}{4}x^{\nu^4}.$
\end{center}
Notice that the fifth term of each new representation also have the same scalar $\frac{1}{4}$.
Now, once the splitting procedure is complete, both $x$ and $y$ have five terms in each representation. Moreover, both lotteries have the same scalar, term to term. }\medskip

Algorithm 2 presented in Appendix \ref{apendice B} is the formalization  of the splitting procedure for two random stable matchings. In Appendix \ref{apendice B}, using the same example, we illustrate the splitting procedure using Algorithm 2 detailed the procedure step by step. Further, the following proposition states that the splitting procedure changes the representation of the two random stable matchings.
\begin{proposition}\label{proposicion reescribir con algoritmo}
Let $x$ and $y$ be two random stable matchings such that $$x=\sum_{\ell=1}^{I}\alpha^0_{\ell}\mu^{x}_{\ell} \text{ ~~~~~   and    ~~~~~} y=\sum^J_{\ell=1}\beta^0_{\ell}\mu^{y}_{\ell}.$$ 
Then, there is $\Omega= \left\{\left(\gamma_\ell,~\tilde{\mu}^x_\ell,~\tilde{\mu}^y_\ell \right): \ell=1,\ldots,\tilde{k}\right\}$ defined by Algorithm 2, where $\tilde{k}$ is the last step of the algorithm  such that
$$x=\sum_{\ell=1}^{\tilde{k} }\gamma_\ell \tilde{\mu}_\ell^{x} \text{ ~~~~~   and    ~~~~~}y=\sum _{\ell=1}^{\tilde{k}}\gamma_\ell \tilde{\mu}_\ell^{y}.$$ 
\end{proposition}
\begin{proof}
See proof in Appendix \ref{apendice B}.
\end{proof}

Once two random stable matchings goes through the splitting procedure, we can define the following domination relation. This domination relation and its equivalence with the partial order defined in Section \ref{order}, are used in the next section to prove the main result.
\begin{definition}
Let $x$ and $y$ be two random stable matchings such that $$x=\sum_{\ell=1}^{\tilde{k} }\gamma_\ell \tilde{\mu}_\ell^{x} \text{ ~~~~~   and    ~~~~~}y=\sum _{\ell=1}^{\tilde{k}}\gamma_\ell \tilde{\mu}_\ell^{y},$$  where for each $\ell=1,\ldots,\tilde{k}$,~~ $0< \gamma_{\ell} \leq 1$, $\sum_{\ell=1}^{\tilde{k}}\gamma_{\ell}=1,$ and for each $\ell=1,\ldots,\tilde{k}-1$ $\tilde{\mu}_\ell^{x} \geq_F \tilde{\mu}_{\ell+1}^{x}$ and $\tilde{\mu}_\ell^{y} \geq_F \tilde{\mu}_{\ell+1}^{y}.$

We say that $\boldsymbol{x}$ \textbf{splittely dominates} $\boldsymbol{y}$ for all firms ($\boldsymbol{x\succeq_{F}^{S} y}$) if
 $\tilde{\mu}^{x}_{\ell} \geq_{F} \tilde{\mu}^{y}_{\ell}$ for each $\ell=1,\ldots,\tilde{k}.$ Analogously, we define $\boldsymbol{x\succeq^S_{W} y}$ for all workers. 
\end{definition}

\begin{remark}
Notice that the partial order $\succeq_F$ ($\succeq_W$) compares two random stable matchings independently if they are or not splitted.
\end{remark}

The following proposition  proves that both domination relations ($\succeq_F$ and $\succeq_F^S$) are equivalent.

\begin{proposition}\label{equivalencia de ordenes}
The partial order $\succeq_F$ is equivalent to the domination relation $\succeq_F^S$.
\end{proposition}
\begin{proof}
See proof in Appendix \ref{apendice B}.
\end{proof}
\begin{corollary}
The domination relation $\succeq_F^S$ is also a partial order.
\end{corollary}

\section{Main result}\label{seccion main result}
Given two random stable matchings represented after the splitting procedure, we define binary operations for random stable matchings to compute the \textit{l.u.b.} and \textit{g.l.b.} for each side of the market. Further, we state the main result of this paper: the set of random stable matchings has a dual lattice structure.

Recall that $\vee_W$, $\wedge_W$, $\vee_F$ and $\wedge_F$ are the binary operations relative to the partial orders $\geq_{W}$ and $\geq_{F}$  defined between two (deterministic) stable matchings. Now, we extend  these binary operations to random stable matchings. Formally,

\begin{definition}\label{defino operaciones binarias}
Let $x$ and $y$ be two random stable matchings such that 
$$x=\sum_{\ell=1}^{\tilde{k}}\gamma_{\ell}\tilde{\mu}^{x}_{\ell} \text{~~~~ and ~~~~}y=\sum_{\ell=1}^{\tilde{k}}\gamma_{\ell}\tilde{\mu}^{y}_{\ell} $$ where for each $\ell=1,\ldots,\tilde{k}$,~~ $0< \gamma_{\ell} \leq 1$, $\sum_{\ell=1}^{\tilde{k}}\gamma_{\ell}=1$,  $\tilde{\mu}^{x}_{\ell},\tilde{\mu}^{y}_{\ell}\in{\mathcal{S(P)}}$, and for each $\ell=1,\ldots,\tilde{k}-1$,~~ $\tilde{\mu}^{x}_{\ell} \geq_{F} \tilde{\mu}^{x}_{\ell+1}$ and  $\tilde{\mu}^{y}_{\ell} \geq_{F} \tilde{\mu}^{y}_{\ell+1}$.

Define $\boldsymbol{x\veebar_F y,~x \barwedge_F y,~x\veebar_W y \text{ and }x \barwedge_W y}$ as follows:
$$\boldsymbol{x\veebar_F y}:=\sum_{\ell=1}^{\tilde{k}}\gamma_{\ell}(\tilde{\mu}^{x}_{\ell}\vee_F \tilde{\mu}^{y}_{\ell}) \text{~~,~~~~~} \boldsymbol{x \barwedge_F y}:=\sum_{\ell=1}^{\tilde{k}}\gamma_{\ell}(\tilde{\mu}^{x}_{\ell} \wedge_F \tilde{\mu}^{y}_{\ell}),$$ 

and
$$\boldsymbol{x\veebar_W y}:=\sum_{\ell=1}^{\tilde{k}}\gamma_{\ell}(\tilde{\mu}^{x}_{\ell}\vee_W \tilde{\mu}^{y}_{\ell}) \text{~~,~~~~~} \boldsymbol{x \barwedge_W y}:=\sum_{\ell=1}^{\tilde{k}}\gamma_{\ell}(\tilde{\mu}^{x}_{\ell} \wedge_W \tilde{\mu}^{y}_{\ell}),$$

\end{definition}

\begin{remark}
It is straightforward by Remark \ref{operaciones en matching es estable} that $x\veebar_F y,~x \barwedge_F y,~x\veebar_W y \text{ and }x \barwedge_W y$ are random stable matchings.
\end{remark}
Now, we are in position to prove that these binary operations defined for random stable matchings are actually the \textit{l.u.b.} and \textit{g.l.b.} for each side of the market.  
\begin{proposition}\label{teorema de operaciones binarias}
 Let $x$ and $y$ be two random stable matchings. Then, for $X\in \{F,W\}$ we have that
$$x\veebar_X y =\textit{l.u.b.}_{\succeq_X}(x,y) \text{ ~~and~~ } x\barwedge_X y =\textit{g.l.b.}_{\succeq_X}(x,y).$$ Also, 
$$x\veebar_F y =x\barwedge_W y  \text{ ~~and~~ } x\veebar_W y =x\barwedge_F y.$$
\end{proposition}
\begin{proof}
See proof in Appendix \ref{apendice B}.
\end{proof}

Now, we are in position to state the main result as follows,
\begin{theorem}\label{teorema principal}
$(\mathcal{RS(P)},\succeq_F,\veebar_F,\barwedge_F)$ and $(\mathcal{RS(P)},\succeq_W,\veebar_W,\barwedge_W)$ are dual lattices. 
\end{theorem}

The following example illustrate how to compute the binary operations for two random stable matchings.

\noindent \textbf{Example 1 (Continued)}\textit{ Given $x$ and $y$ represented as in Proposition \ref{proposicion reescribir con algoritmo}, we compute $x\veebar_F y$ and  $x\barwedge_F y$ as follows: (The other two cases are similar)}
\begin{center}
$
x=\frac{1}{6}x^{\nu_1}+\frac{1}{12}x^{\nu_1}+\frac{5}{12}x^{\nu_2}+\frac{1}{12}x^{\nu_2}+\frac{1}{4}x^{\nu_4},
$
\end{center}
\begin{center}
$y=\frac{1}{6}x^{\nu_1}+\frac{1}{12}x^{\nu_3}+\frac{5}{12}x^{\nu_3}+\frac{1}{12}x^{\nu_4}+\frac{1}{4}x^{\nu_4}.
$
\end{center}
\begin{center}
$
x\veebar_F y=\frac{1}{6}x^{\nu_1 \vee_F \nu_1}+\frac{1}{12}x^{\nu_1 \vee_F \nu_3}+\frac{5}{12}x^{\nu_2 \vee_F \nu_3}+\frac{1}{12}x^{\nu_2 \vee_F \nu_4}+\frac{1}{4}x^{\nu_4 \vee_F \nu_4} 
$
\end{center}
\begin{center}
$
=\frac{1}{6}x^{\nu_1}+\frac{1}{12}x^{\nu_1}+\frac{5}{12}x^{\nu_1}+\frac{1}{12}x^{\nu_2}+\frac{1}{4}x^{\nu_4},
$
\end{center}
\begin{center}
$
=\frac{2}{3}x^{\nu_1}+\frac{1}{12}x^{\nu_2 }+\frac{1}{4}x^{\nu_4 }. 
$
\end{center}
\begin{center}
$
x\barwedge_F y=\frac{1}{6}x^{\nu_1 \wedge_F \nu_1}+\frac{1}{12}x^{\nu_1 \wedge_F \nu_3}+\frac{5}{12}x^{\nu_2 \wedge_F \nu_3}+\frac{1}{12}x^{\nu_2 \wedge_F \nu_4}+\frac{1}{4}x^{\nu_4 \wedge_F \nu_4}
$
\end{center}
\begin{center}
$
=\frac{1}{6}x^{\nu_1}+\frac{1}{12}x^{\nu_3}+\frac{5}{12}x^{\nu_4}+\frac{1}{12}x^{\nu_4}+\frac{1}{4}x^{\nu_4},
$
\end{center}
\begin{center}
$
=\frac{1}{6}x^{\nu_1}+\frac{1}{12}x^{\nu_3 }+\frac{3}{4}x^{\nu_4}.
$
\end{center}

\subsection{Binary operations for rational random stable matchings}\label{rational}

In this subsection, we compute the \textit{g.l.b.} and \textit{l.u.b.} for two random stable matchings where each scalar of the lottery is a \textit{rational number}. These random stable matchings are called \textbf{rational random stable matchings}. In this case, the splitting procedure is directly and different to the procedure described by Algorithm 2.

Let $x$ and $y$ be two rational random stable matchings, represented as follows:

\begin{equation}\label{alpha racional}
 x=\sum_{i=1}^{I}\alpha_{i}\mu^{x}_{i},
\end{equation}
 such that, $0< \alpha_{i} \leq 1$ for each $i=1,\ldots,I$; $\sum_{i=1}^{I}\alpha_{i}=1$, $\mu^{x}_{i}\in{\mathcal{S(P)}}$, $\alpha_{i}$ is a rational number and $\mu^{x}_{i} >_{F} \mu^{x}_{i+1}$; where  $i=1,...,I-1.$
\begin{equation}\label{beta racional}
 y=\sum_{j=1}^{J}\beta_{j}\mu^{y}_{j},
\end{equation}
such that, $0< \beta_{j} \leq 1$ for each $j=1,\ldots,J$; $\sum_{j\in{J}}\beta_{j}=1$, $\mu^{y}_{j}\in{\mathcal{S(P)}}$, $\beta_{j}$ is a rational number and $\mu^{y}_{j} >_{F} \mu^{y}_{j+1}$; where $j=1,...,J-1$.

Since $\alpha_i$ and $\beta_j$ are positive rational numbers, we have that for each $\alpha_i$ there are natural numbers $a_i,~b_i$ such that $\alpha_i=\frac{a_i}{b_i}$. Similarly, for each $\beta_j$ there are natural numbers $c_j,~d_j$ such that $\beta_j=\frac{c_j}{d_j}.$

Denote by $e$ the \textit{least common multiple} (\textbf{\textit{lcm}}) of all denominators $b_i,~d_j$ for each $i=1,\ldots, I$ and for each $j=1,\ldots, J$. That is,
$$
e=\textit{lcm}(b_1,\ldots,b_I,d_1,\ldots,d_J).
$$

Then, we can write $\alpha_i=\frac{a_i}{b_i}=\frac{a_i \frac{e}{b_i}}{e}$ and $\beta_i=\frac{c_j}{d_j}=\frac{c_j \frac{e}{d_j}}{e}$ for each $i=1,\ldots, I$ and for each $j=1,\ldots, J.$ Hence, we can write all the scalars $\alpha$ and $\beta$ with the same denominator.

Denote by $\gamma_k=\frac{1}{e}$ and define 

$$
\boldsymbol{\tilde{\mu}_k^{x}}:=\left\{
\begin{array}{ll}
\mu^x_{1} & \text{~for~} k=1,\ldots,\frac{a_1}{b_1}e\\
\mu^x_{2} & \text{~for~} k=\frac{a_1}{b_1}e+1,\ldots,\left(\frac{a_2}{b_2}+\frac{a_1}{b_1}\right) e\\
\vdots & ~~~\vdots\\
\mu^x_{I} & \text{~for~} k=\left(\displaystyle\sum_{n=1}^{I-1}\frac{a_{n}}{b_{n}}\right) e+1,\ldots,\left(\displaystyle\sum_{n=1}^{I}\frac{a_{n}}{b_{n}}\right) e\\
\end{array}
\right.
$$

$$
\boldsymbol{\tilde{\mu}_k^{y}}:=\left\{
\begin{array}{ll}
\mu^y_{1} & \text{~for~} k=1,\ldots,\frac{c_1}{d_1}e\\
\mu^y_{2} & \text{~for~} k=\frac{c_1}{d_1}e+1,\ldots,\left(\frac{c_2}{d_2}+\frac{c_1}{d_1}\right) e\\
\vdots & ~~~\vdots\\
\mu^y_{J} & \text{~for~} k=\left(\displaystyle\sum_{m=1}^{J-1}\frac{c_{m}}{d_{m}}\right) e+1,\ldots,\left(\displaystyle\sum_{m=1}^{J}\frac{c_{m}}{d_{m}}\right) e\\
\end{array}
\right.
$$
Then, we have that 
\begin{equation}\label{equacion racional x}
x=\sum_{i=1}^{I}\alpha_{i}\mu^{x}_{i}=\sum_{i=1}^{I}\frac{a_i}{b_i}\mu^{x}_{i}=\sum_{i=1}^{I}\frac{a_i \frac{e}{b_i}}{e}\mu^{x}_{i}=\sum_{k=1}^{e}\frac{1}{e}\tilde{\mu}^x_k.
\end{equation}
Analogously, we have that 
\begin{equation}\label{equacion racional y}
y=\sum_{j=1}^{J}\beta_{j}\mu^{y}_{j}=\sum_{j=1}^{J}\frac{c_j}{d_j}\mu^{y}_{j}=\sum_{j=1}^{J}\frac{c_j \frac{e}{d_j}}{e}\mu^{y}_{j}=\sum_{k=1}^{e}\frac{1}{e}\tilde{\mu}^y_k.
\end{equation}

Given two rational random stable matchings $x$ and $y$, represented as in (\ref{equacion racional x}) and (\ref{equacion racional y}), to compute $x\veebar_F y$, $x \barwedge_F y$, $x\veebar_W y$ and $x \barwedge_W y$ we state the following corollary of Theorem \ref{teorema de operaciones binarias}.
\begin{corollary}\label{corolario para rational}
Let $x$ and $y$ be two rational random stable matchings (i.e. each $\alpha$ and each $\beta$ in (\ref{alpha racional}) and (\ref{beta racional}) are rational numbers). Then, for $X\in \{F,W\}$ we have that
$$x\veebar_X y=\sum_{k=1}^e\frac{1}{e}(\tilde{\mu}^{x}_{k}\vee_X \tilde{\mu}^{y}_{k}) \text{~~~~  and ~~~~} x \barwedge_X y=\sum_{k=1}^e\frac{1}{e}(\tilde{\mu}^{x}_k \wedge_X \tilde{\mu}^{y}_k).$$ 
\end{corollary}

\noindent \textbf{Example 1 (Continued)}\textit{ Let $x$ and $y$ be two random stable matchings represented as in Proposition \ref{teorema del orden},}
\begin{center}
$x=\frac{1}{4}x^{\nu_1}+\frac{1}{2}x^{\nu_2}+\frac{1}{4}x^{\nu_4},
$
\end{center}
\begin{center}
$
y=\frac{1}{6}x^{\nu^1}+\frac{1}{2}x^{\nu^3}+\frac{1}{3}x^{\nu^4}.
$
\end{center}

\textit{Let $e=\textit{lcm}(2,3,4,6)=12$. Then, the random stable matchings $x$ and $y$ can be represented as:}

\begin{center}
$
x=\frac{1}{12}x^{\nu_1}+\frac{1}{12}x^{\nu_1}+\frac{1}{12}x^{\nu_1}+\frac{1}{12}x^{\nu_2}+\frac{1}{12}x^{\nu_2}+\frac{1}{12}x^{\nu_2}+\frac{1}{12}x^{\nu_2}+\frac{1}{12}x^{\nu_2}
$
\end{center}
\begin{center}
$+\frac{1}{12}x^{\nu_2}+\frac{1}{12}x^{\nu_4}+\frac{1}{12}x^{\nu_4}+\frac{1}{12}x^{\nu_4},
$
\end{center}
\begin{center}
$
y=\frac{1}{12}x^{\nu_1}+\frac{1}{12}x^{\nu_1}+\frac{1}{12}x^{\nu_3}+\frac{1}{12}x^{\nu_3}+\frac{1}{12}x^{\nu_3}+\frac{1}{12}x^{\nu_3}+\frac{1}{12}x^{\nu_3}+\frac{1}{12}x^{\nu_3}
$
\end{center}
\begin{center}
$
+\frac{1}{12}x^{\nu_4}+\frac{1}{12}x^{\nu_4}+\frac{1}{12}x^{\nu_4}+\frac{1}{12}x^{\nu_4}.
$
\end{center}
\textit{Then, }

\begin{center}
$
x\veebar_F y=\frac{1}{12}x^{\nu_1}+\frac{1}{12}x^{\nu_1}+\frac{1}{12}x^{\nu_1}+\frac{1}{12}x^{\nu_1}+\frac{1}{12}x^{\nu_1}+\frac{1}{12}x^{\nu_1}+\frac{1}{12}x^{\nu_1}
$
\end{center}
\begin{center}
$
+\frac{1}{12}x^{\nu_1}+\frac{1}{12}x^{\nu_2}+\frac{1}{12}x^{\nu_4}+\frac{1}{12}x^{\nu_4}+\frac{1}{12}x^{\nu_4}
$
\end{center}
\begin{center}
$
=\frac{2}{3}x^{\nu_1}+\frac{1}{12}x^{\nu_2 }+\frac{1}{4}x^{\nu_4 }. 
$
\end{center}
\textit{Analogously for $x\barwedge_F y,~x\veebar_W y$ and $x\barwedge_W y.$}
\section{Concluding remarks}\label{conclusiones}
In this paper, we prove an important result that involves two very much studied topics in the matching literature: random stable matchings and lattice structure. The many-to-many matching markets with substitutable preferences satisfying the L.A.D. are the most general matching markets in which it is known that the binary operations between two stable matchings (\textit{l.u.b.} and \textit{g.l.b.}) are computed via pointing functions.  For these markets, we prove that the set of random stable matchings endowed with a partial order has a dual lattice structure. Moreover, we present natural binary operations to compute \textit{l.u.b.} and \textit{g.l.b.} between two random stable matching for each side of the markets.  The partial order defined in this paper is a generalization of the first-order stochastic dominance for the case in which agents have substitutable preferences satisfying the L.A.D..
For more general matching markets, for instance markets that only satisfy substitutability (not L.A.D.), the binary operations between (deterministic) stable matchings are computed as fixed points. Then, the lattice structure of the set of random stable matchings for these markets is still an open problem, left for future research.

\newpage
\section*{Appendix}
\appendix

\section{The decreasing representation}\label{apendice A}
The following technical results are used in the proof of Theorem \ref{teorema del orden}.
\begin{lemma}\label{mu k en Bk}
$\mu_k\in B_k$ for each $k=1,\ldots,\tilde{k}$.
\end{lemma}
\begin{proof}
Let $\mu_1=\displaystyle{\bigvee_{\nu\in B_1}}\displaystyle{_{\!\!\!\!{F}}}~\nu$. Then, by definition of $B_1,$ we have that $\mu_1 \in B_1.$ Let $\mu_2=\bigvee_{\nu\in B_2}\nu$ and  assume that $\mu_2 \notin B_2$, then $\mu_2 \in C_2$. Since  $\mu_2$ is computed via pointing functions, there is $\nu'\in B_2$ such that $\nu'\in C_1$, which is a contradiction, then $\mu_2\in B_2$. Similar arguments proves that $\mu_k \in B_k$ for each $k=1, \ldots,\tilde{k}$, where $\tilde{k}$ is the last step of Algorithm 1. 
\end{proof}

\begin{lemma}\label{inclusion de Bk}
If $B_k \neq \emptyset$, then $B_{k+1} \subset B_k$.
\end{lemma}
\begin{proof}
By definition of $\mu_k$ and $C_k$, $\mu_k\in B_k \cap C_k.$ Then, $B_{k+1}=B_k \setminus C_k \subset B_k.$
\end{proof}

\begin{lemma}\label{nu tilde en el ultimo de los Bk}
Let $\tilde{\nu}=\displaystyle\bigwedge_{\nu\in B_1}\displaystyle{_{\!\!\!\!{F}}}~\nu,$ and ${\tilde{k}}$ the step of Algorithm 1 in which $B_{{\tilde{k}}} \neq \emptyset$ and $B_{{\tilde{k}}+1} =\emptyset.$ Then, $\tilde{\nu}\in B_{{\tilde{k}}}.$ 
\end{lemma}
\begin{proof}
Let $\tilde{\nu}=\displaystyle\bigwedge_{\nu\in B_1}\displaystyle{_{\!\!\!\!{F}}}~\nu$ and ${\tilde{k}}$ the step of Algorithm 1 in which $B_{{\tilde{k}}} \neq \emptyset$ and $B_{{\tilde{k}}+1} =\emptyset.$ By definition of $B_1$, $\tilde{\nu}\in B_1$. Assume that $\tilde{\nu} \notin B_{{\tilde{k}}}$, then there is a $k'<{\tilde{k}}$ such that $\tilde{\nu}\in B_{k'}$ and $\tilde{\nu}\notin B_{k'+1}$. Then, $\tilde{\nu}\in C_{k'}$. Hence, by definition of $C_{k'}$, there is a pair $(i',j')$ such that $x_{i',j'}^{k'}=\alpha_{k'}$ and $x_{i',j'}^{\tilde{\nu}}=x_{i',j'}^{\mu_{k'}}=1.$ Notice that, by definition of $\mu_{k'}$, we have that $j'\in \mu_{k'}(i')=Ch(\cup_{\nu\in B_{k'}}\nu(i'), >_{i'})$. Since the preferences relation $>_{i'}$ is substitutable and $\tilde{\nu} \in B_{k'}$, we have that
\begin{equation}\label{ecuacion contradiccion}
j'\in Ch(\tilde{\nu}(i')\cup \{j'\}, >_{i'}).
\end{equation}
By Lemma \ref{inclusion de Bk} and $k'<{\tilde{k}}$, we have that $B_{k'+1}\neq \emptyset.$ Then, there is $\nu'\in B_{k'}$ such that $\nu'\notin C_{k'}.$ We claim that $j' \notin \nu'(i')$. If $j' \in \nu'(i')$, for $(i',j')$ we have $x_{i',j'}^{k'}=\alpha_{k'}$ and $x_{i',j'}^{\nu^{{\tilde{k}}}}=x_{i',j'}^{\mu_{k'}}=1,$ then $\nu' \in C_{k'}$, which is a contradiction. Thus, $j' \notin \nu'(i')$. Since $\nu'\in B_{k'}\subseteq B_1 $, then $\nu'\geq_F \tilde{\nu}.$ That is, $\nu'(i')=Ch(\tilde{\nu}(i') \cup \nu'(i'),>_{i'}).$ Now, given that $j' \in \tilde{\nu}(i') \setminus \nu'(i')$, we have that  $j'\notin Ch(\tilde{\nu}(i')\cup \{j'\}, >_{i'}), $ which is a contradiction with (\ref{ecuacion contradiccion}). Therefore,   $\tilde{\nu}\in B_{{\tilde{k}}}.$ 
\end{proof}

\begin{lemma}\label{todos suman lo mismo}
Let $\mu \in \mathcal{S(P)}$, $x^1$ a random stable matching, and $x^{k}=\frac{x^{k-1}-\alpha_{k-1}x^{\mu_{k-1}}}{1-\alpha_{k-1}}$ be the matrix construct by Algorithm 1 in Step $k$. Then, for each $k$, we have that $\sum_{i\in F}x^k_{i,j}=|\mu(j)|$ for each $j\in W$, and $\sum_{j\in W}x^k_{i,j}=|\mu(i)|$ for each $i\in F$.

\end{lemma}
\begin{proof}
Let $\mu \in \mathcal{S(P)}$ and let $k=1$ the first step of Algorithm 1. If $B_2= \emptyset$, then $B_1=C_1$. That is, $\tilde{\nu}\in C_1$. Hence, there is $(i,j)\in \mathcal{L}_1$ such that $x_{i,j}^{\mu_1}=1, x_{i,j}^{\tilde{\nu}}=1$ and $x_{i,j}^1=\alpha_1$. Then, for each $\nu\in B_1$ such that $\mu_1\geq_F \nu \geq_F \tilde{\nu}$ we have that $x_{i,j}^{\nu}=1$. Hence, $\alpha_1=1$. Since $supp(x^{\mu_1})\subseteq supp(x^1)$ and $\alpha_1=min\{x^1_{i,j} :x^{\mu_{1}}_{i,j}=1 \}$, then $x^{1}=x^{\mu_1}$. Thus, by Proposition (RHT) and definition of incidence vector, we have that $\sum_{i\in F}x^{\mu_1}_{i,j}=|\mu(j)|$ for each $j\in W$, and $\sum_{j\in W}x^{\mu_1}_{i,j}=|\mu(i)|$ for each $i\in F$. 

Assume that $B_2 \neq \emptyset$ and $\sum_{i\in F}x^{k-1}_{i,j}=|\mu(j)|$ for each $j\in W$, then by Proposition (RHT) and  definition of $x^k$ we have that 
$$
\sum_{i\in F} x^{k}_{i,j}=\frac{\sum_{i\in F}x^{k-1}_{i,j}-\alpha_{k-1}\sum_{i\in F}x^{\mu_{k-1}}_{i,j}}{1-\alpha_{k-1}}=\frac{|\mu(j)|-\alpha_{k-1}|\mu(j)|}{1-\alpha_{k-1}}=|\mu(j)|.
$$
Therefore, $\sum_{i\in F}x^k_{i,j}=|\mu(j)|$ for each $j\in W$ and for each $k=1,\ldots,\tilde{k}$. Similarly, we can prove that $\sum_{j\in W}x^k_{i,j}=|\mu(i)|$ for each $i\in F$ and for each $k=1,\ldots,\tilde{k}$.
\end{proof}

\begin{lemma}\label{Bk+1 no vacio, si y solo si alpha k <1}
$B_{k+1}\neq \emptyset$ if and only if $\alpha_k<1$.
\end{lemma}
\begin{proof}
$\Longrightarrow)$ Let $B_{k+1}\neq \emptyset$, then $B_k \neq C_k$. Hence, $|B_k|>1$. By Lemma \ref{nu tilde en el ultimo de los Bk} we have that $\tilde{\nu} \in B_k$. Also, by definition of $\mu_k$, we have that $\tilde{\nu}\neq \mu_k$. Then, by Proposition (RHT), there are at least three agents $i'\in F$ and $\tilde{j},j' \in W$ such that:

$$
x^k_{i',j'}>0,~~x^k_{i',\tilde{j}}>0,~~x^{\mu_k}_{i',j'}=1,~~x^{\mu_k}_{i',\tilde{j}}=0,~~x^{\tilde{\nu}}_{i',j'}=0,\text{ and }~x^{\tilde{\nu}}_{i',\tilde{j}}=1.
$$
By Lemma \ref{todos suman lo mismo}, we have that $\sum_{j\in W}x^k_{i',j}=|\mu_k(i')|=|\tilde{\nu}(i')|$. Since $supp(x^{\mu_k})\subset supp(x^k)$, and $supp(x^{\tilde{\nu}})\subset supp(x^k)$, we have that $|\{j\in W: x^{k}_{i',j}>0\}|>|\mu_k(i')|$. Then there is an agent $\hat{j}\in W$ such that $x_{i',\hat{j}}^{\mu_k}=1$ and $0<x^k_{i',\hat{j}}<1.$ Thus, $\alpha_k=min\{ x^k_{i,j}:x^{\mu_k}_{i,j}=1\} \leq x^k_{i',\hat{j}} <1.$

$\Longleftarrow)$ Let $\alpha_k <1.$ Then there is a pair $(i',j')$ such that $x_{i',j'}^{\mu_k}=1$ and $x^k_{i',j'}=\alpha_k<1$.
Then, by Lemma \ref{todos suman lo mismo} there is a pair $(i',\tilde{j})$ such that $x_{i',\tilde{j}}^{\mu_k}=0$, $x_{i',\tilde{j}}^k >0$ and $x_{i',\tilde{j}}^{ \tilde{\nu}}=1$. Hence, for each pair $(i,j)$ such that $x_{i,j}^{\mu_k}=1$ and $x_{i,j}^{\tilde{\nu}}=1$, by Proposition (RHT) we have that $x^{\nu}_{i,j}=1$ for each $\nu \in B_k$. Thus,  $x_{i,j}^{k}\geq x_{i',\tilde{j}}^{k}+x_{i',j'}^{k}=x_{i',\tilde{j}}^{k}+\alpha_k>\alpha_k$,   Then, $\tilde{\nu}\notin C_k$. Therefore, $B_{k+1}\neq \emptyset.$
\end{proof}

\begin{corollary}\label{alpha k =1 entonces matching estable}
If $\alpha_k=1$, then $x^{k}=x^{\mu_k}.$
\end{corollary}
\begin{proof}
Let $\mathcal{L}_k=\{(i,j)\in F\times W: x^k_{i,j}=\alpha_k \text{ and }x_{i,j}^{\mu_k}=1\}$ and recall that by definition of $\mu_k$, we have that $supp(x^{\mu_k})\subseteq supp(x^{k})$. If $\alpha_k=1$, then $ \mathcal{L}_k=supp(x^{\mu_k})$. By Lemma \ref{todos suman lo mismo}, we have that $\sum_{i\in F}x^k_{i,j}=|\mu_{k}(j)|$ for each $i\in F$, then $supp(x^k)=supp(x^{\mu_k})$. Therefore, $x^{k}=x^{\mu_k}.$ 
\end{proof}

\begin{proof}[Proof of Theorem \ref{teorema del orden}]
Let $x$ be a random stable matching. The output of Algorithm 1 are the sets $\mathcal{M}=\{\mu_1,\ldots, \mu_{\tilde{k}}\}$ and  $\Lambda=\{\alpha_1,\ldots,\alpha_{\tilde{k}}\}$.

By Lemma \ref{inclusion de Bk}, we have that $B_{k+1} \subset B_k$. By the finiteness of the set of stable matchings, we have that there is a step of Algorithm 1, say Step $\tilde{k}$, such that $B_{\tilde{k}+1}=\emptyset$. Then, the algorithm stops. Hence, by Lemma \ref{Bk+1 no vacio, si y solo si alpha k <1} we have that $\alpha_{\tilde{k}}=1$. Therefore, by Corollary \ref{alpha k =1 entonces matching estable}, we have that $x^{\tilde{k}}=x^{\mu_{\tilde{k}}}$. 

Recall that  $\mu_k= \displaystyle\bigvee_{\nu \in B_k}\displaystyle{_{\!\!\!\!{F}}}~~\nu,$ and $\mu_{k+1}=\displaystyle\bigvee_{\nu \in B_{k+1}}\displaystyle{_{\!\!\!\!\!\!\!{F}}}~~\nu$. By Lemma \ref{inclusion de Bk} we have that  $B_{k+1}\subset B_k$, by Lemma \ref{mu k en Bk} we have that $\mu_k \in B_k$, and by definition of $C_k$ we have that $\mu_k \in C_k$. Hence, $\mu_k \notin B_{k+1}$. Then, $\mu_k >_F \mu_{k+1}.$

Let 
$ \beta_{1}=\alpha_{1},~\beta_{2}=(1-\alpha_{1})\alpha_{2},~\beta_{3}=(1-\alpha_{1})(1-\alpha_{2})\alpha_{3},\ldots,$ and $\beta_{\tilde{k}}= \prod_{k=1}^{\tilde{k}-1}(1-\alpha_k).$

Now, we prove that $\sum_{k=1}^{\tilde{k}}\beta_{k}=1.$ 
$$\sum_{k=1}^{\tilde{k}}\beta_{k}=\sum_{k=1}^{\tilde{k}-1}\beta_{k}+\beta_{\tilde{k}}=\sum_{k=1}^{\tilde{k}-1} \prod_{\ell=1}^{k-1}(1-\alpha_\ell)+\prod_{\ell=1}^{\tilde{k}-1}(1-\alpha_\ell).$$
Note that, 
$$\beta_{\tilde{k}-1}+\beta_{\tilde{k}}=\prod_{\ell=1}^{\tilde{k}-2}(1-\alpha_\ell)\alpha_{\tilde{k}-1}+\prod_{\ell=1}^{\tilde{k}-1}(1-\alpha_\ell)=\prod_{\ell=1}^{\tilde{k}-2}(1-\alpha_\ell)(\alpha_{\tilde{k}-1}+(1-\alpha_{\tilde{k}-1}))=\prod_{\ell=1}^{\tilde{k}-2}(1-\alpha_\ell).$$
Also, we have that
$$\beta_{\tilde{k}-2}+\beta_{\tilde{k}-1}+\beta_{\tilde{k}}=\prod_{\ell=1}^{\tilde{k}-3}(1-\alpha_\ell)\alpha_{\tilde{k}-2}+\prod_{\ell=1}^{\tilde{k}-3}(1-\alpha_\ell)(1-\alpha_{\tilde{k}-2})=\prod_{\ell=1}^{\tilde{k}-3}(1-\alpha_\ell).$$
Continuing this inductive process,
$\beta_2+ \cdots +\beta_{\tilde{k}}=(1-\alpha_1).$ Then,
$$
\sum_{k=1}^{\tilde{k}}\beta_k=\beta_1 +\sum_{k=2}^{\tilde{k}}\beta_k=\alpha_1+(1-\alpha_1)=1.$$

Therefore, $$x=\sum_{k=1}^{\tilde{k}}\beta_k x^{\mu_k}$$
where $0<\beta_k \leq 1$, $\sum_{k=1}^{\tilde{k}}\beta_k=1$, and  $\mu_{k}>_{F} \mu_{k+1}$ for each $k=1,\ldots,\tilde{k}-1.$  

\textbf{Uniqueness:}
Assume that $x$ has two different representations:
$$
x=\sum_{\nu\in A} \lambda_{\nu} x^{\nu}=\sum_{\nu'\in A'} \lambda'_{\nu'} x^{\nu'}
$$
where $0< \lambda_{\nu} \leq 1,~0< \lambda'_{\nu'} \leq 1,~\sum_{\nu\in{A}}\lambda_{\nu}=1,~\sum_{\nu'\in{A'}}\lambda'_{\nu'}=1,~ \mbox{ and } \nu,\nu'\in{\mathcal{S(P)}}$.

Since, $\bigcup_{\nu\in A}\nu(i)=\{j:x_{i,j}>0\}=\bigcup_{\nu'\in A'}\nu'(i)$  then,  $\mu_1(i)=Ch(\bigcup_{\nu\in B_1}\nu(i), >_i)=\\Ch(\bigcup_{\nu'\in B'_1}\nu'(i), >_i)=\mu'_1(i)$ for each $i \in F$. Therefore, $\mu_1=\mu'_1.$

Let $k>1$ such that $\mu_1=\mu'_1, \ldots,\mu_{k-1}=\mu'_{k-1}$. 
Then, 
$x^k=\frac{x^{k-1}-\alpha_{k-1}x^{\mu_{k-1}}}{1-\alpha_{k-1}}=\frac{x^{k-1}-\alpha_{k-1}x^{\mu'_{k-1}}}{1-\alpha_{k-1}}.$ 

We claim that 
$\{(i,j):x^{k}_{i,j}>0\}=\{(i,j):\bigcup_{\nu\in B_k}x^{\nu}_{i,j}=1\}$ (and  $\{(i,j):x^{k}_{i,j}>0\}=\{(i,j):\bigcup_{\nu'\in B'_k}x^{\nu'}_{i,j}=1\}$). If not, there is a pair $(i,j)$ such that $x^k_{i,j}>0$ and $x^{\nu}_{i,j}=0$ for each $\nu \in B_k$. Then, $x_{i,j}^{\tilde{k}}>0$ and there is no $\nu \in B_{\tilde{k}}$ such that $x^{\nu}_{i,j}=0$. This contradicts that $x^{\tilde{k}}_{i,j}\neq x^{{\mu}_{\tilde{k}}}_{i,j}.$ Hence, $\{(i,j):\bigcup_{\nu\in B_k}x^{\nu}_{i,j} =1\} \supseteq \{(i,j):x^{k}_{i,j}>0\}$.

Assume that there is $\nu \in B_k$ such that $x^{\nu}_{i,j} =1$ and $x^{k}_{i,j}=0$. Since $x^{\nu}_{i,j}=1,$ then $x_{i,j}>0$. Hence, there is $k'<k$ such that $x_{i,j}^{k'}>0$ and  $x_{i,j}^{k'+1}=0$. Since $x^{k'+1}_{i,j}=\frac{x^{k'}_{i,j}-\alpha_{k'}x^{\mu_{k'}}_{i,j}}{1-\alpha_{k'}}=0$ then, $x^{k'}_{i,j}=\alpha_{k'}x^{\mu_{k'}}_{i,j}$. Hence, $(i,j)\in \mathcal{L}_{k'}$ and $\nu \in C_{k'}$ because we assume that $x^{\nu}_{i,j}=1$. Thus, $\nu\notin B_{k'+1} \supseteq B_k$ and this implies that $x^{k}_{i,j}>0$, which is a contradiction. Therefore, $\{(i,j):\bigcup_{\nu\in B_k}x^{\nu}_{i,j} =1\} \subseteq \{(i,j):x^{k}_{i,j}>0\}$.

Similar arguments prove that $\{(i,j):x^{k}_{i,j}>0\}=\{(i,j):\bigcup_{\nu'\in B'_k}x^{\nu'}_{i,j}=1\}.$

Since $\bigcup_{\nu\in B_k}\nu(i)=\{j:x^{k}_{i,j}>0\}=\bigcup_{\nu'\in B'_k}\nu'(i)$, then
$\mu_{k}(i)=Ch(\bigcup_{\nu\in B_k}\nu(i), \geq_i)=Ch(\{j:x^{k}_{i,j}>0\}, >_i)=Ch(\bigcup_{\nu'\in B'_k}\nu'(i), >_i)=\mu'_{k}(i)$ for each $i\in F.$ Therefore, $\mu_k=\mu'_k.$ 

\end{proof}

\begin{proof}[Proof of Proposition \ref{hospital rural para random}]
Let $\mu\in \mathcal{S(P)}$. By Lemma \ref{todos suman lo mismo} of Appendix \ref{apendice A}, for the case in which $k=1$, $x=x^{k}$, we have that $\sum_{i\in F}x_{i,j}=|\mu(j)|$ for each $j\in W$. Analogously, we have that $\sum_{i\in F}x'_{i,j}=|\mu(j)|$ for each $j\in W$. Hence,
$$\sum_{i\in F}x_{i,j}=\sum_{i\in F}x'_{i,j}.$$
 
The proof of $\sum_{j\in W}x_{i,j}=\sum_{j\in W}x'_{i,j}$ for each $i\in F$ is analogously.
\end{proof}

\newpage

\section{Partial order and the splitting procedure}\label{apendice B}

\subsection*{Proof of partial order}
\begin{proof}[Proof of Proposition \ref{Es orden parcial}]
Let $x$, $y$ and $z$ be random stable matchings represented as
$$x=\sum_{i=1}^{I}\alpha_{i}x^{\mu^{x}_{i}},~y=\sum_{j=1}^{J}\beta_{j}x^{\mu^{y}_{j}} \text{ and } z=\sum_{k=1}^{K}\gamma_{k}x^{\mu^{z}_{k}}.$$

\noindent \textbf{Reflexivity: $\boldsymbol{x\succeq_{F}x}$.}

By the uniqueness of the representation of $x$ following Theorem \ref{teorema del orden}, we have that for each $\mu^x_{k}$ 
$$\sum_{i:\mu_{i}^{x}\geq_{F}\mu_{k}^{x}}\alpha_{l} \geq \sum_{i:\mu_{i}^{x}\geq_{F}\mu_{k}^{x}}\alpha_{l}.$$

\noindent \textbf{Transitivity: If $\boldsymbol{x\succeq_{F}y}$ and $\boldsymbol{y\succeq_{F}z}$, then $\boldsymbol{x\succeq_{F}z}$. }

Since $y\succeq_{F}z$, then $\sum_{l:\mu_{l}^{y}\geq_{F}\mu_{k}^{z}}\beta_{l} \geq \sum_{n:\mu_{n}^{z}\geq_{F}\mu_{k}^{z}}\gamma_{n}$ for each $\mu^z_{k}.$
Since $x\succeq_{F}y$, then $\sum_{i:\mu_{i}^{x}\geq_{F}\mu_{j}^{y}}\alpha_{i} \geq \sum_{l:\mu_{l}^{y}\geq_{F}\mu_{j}^{y}}\beta_{l}$ for each $\mu^y_{j}$.
Recall that  $x$, $y$ and  $z$  are represented following Theorem \ref{teorema del orden}. Then, for each $\mu_k^{z}$ there is an unique $\mu^{y}_j=min_{\geq_F}\{\mu_l^y: \mu_l^y\geq_F \mu_k^z \}$ such that 
\bigskip
\begin{center}

\begin{tabular}{|l|l|}\hline
$\displaystyle{\sum_{m:\mu_{m}^{x}\geq_{F}\mu_{k}^{z}}}\alpha_{m}=\sum_{i:\mu_{i}^{x}\geq_{F}\mu_{j}^{y}}\alpha_{i}$& by $\{\mu_l^{y}:\mu_l^{y} \geq_F \mu_j^{y}\}=\{\mu_l^{y}:\mu_l^{y}\geq_F \mu_k^{z}\}$\\\hline
$\displaystyle{\sum_{i:\mu_{i}^{x}\geq_{F}\mu_{j}^{y}}}\alpha_{i} \geq\sum_{l:\mu_{l}^{y}\geq_{F}\mu_{j}^{y}}\beta_{l}$& by  $x\succeq_{F}y$\\\hline
$\displaystyle{\sum_{l:\mu_{l}^{y}\geq_{F}\mu_{j}^{y}}}\beta_{l}=\sum_{l:\mu_{l}^{y}\geq_{F}\mu_{k}^{z}}\beta_{l}$& by  $\{\mu_m^{x}:\mu_m^{x} \geq_F \mu_k^{z}\}=\{\mu_m^{x}:\mu_m^{x}\geq_F \mu_j^{y}\}$\\\hline
$\displaystyle{\sum_{l:\mu_{l}^{y}\geq_{F}\mu_{k}^{z}}}\beta_{l} \geq \sum_{n:\mu_{n}^{z}\geq_{F}\mu_{k}^{z}}\gamma_{n}$& by  $y\succeq_{F}z$\\\hline
\end{tabular}
\end{center}
Hence, for each $\mu^z_{k}$ we have that $$\sum_{m:\mu_{m}^{x}\geq_{F}\mu_{k}^{z}}\alpha_{m} \geq \sum_{n:\mu_{n}^{z}\geq_{F}\mu_{k}^{z}}\gamma_{n}.$$
Therefore,  $x\succeq_{F}z.$

\noindent \textbf{Antisymmetry: If $\boldsymbol{x\succeq_{F}y}$ and $\boldsymbol{y\succeq_{F}x}$, then $\boldsymbol{x=y}$.}

Assume that $x\succeq_{F}y$ and $x \neq y$, then we prove that $y\nsucceq_{F}x$. By definition of $x\succeq_{F}y$ we have that $x\succeq_{f}y$ for each $f\in F$. Since $x\neq y$, then there is at least one $f'\in F$ such that $x\succ_{f'}y$. Hence, by definition of $x\succ_{f'}y$, there is $\mu_j^{y}(f')$ such that $$\sum_{i:\mu_{i}^{x}(f')\geq_{f'}\mu_{j}^{y}(f')}\alpha_{i} > \sum_{l:\mu_{l}^{y}(f')\geq_{f'}\mu_{j}^{y}(f')}\beta_{l}.$$ Then, $y\nsucceq_{f'}x$, which in turns implies that $y\nsucceq_{F}x$.

Therefore, the domination relation $\succeq_{F}$ is a partial order.
\end{proof}

\subsection*{Algorithm 2}
Let $x$ and $y$ be two random stable matchings such that, 
 $$x=\sum_{i=1}^{I}\alpha^0_{i}\mu^{x}_{i} \text{ ~~~~~   and    ~~~~~} y=\sum^J_{j=1}\beta^0_{j}\mu^{y}_{j}.$$
where  $0<\alpha^{0}_i \leq I \text{ for } i=1,\ldots,I,~ 0<\beta^{0}_j \leq J \text{ for }j=1,\ldots,J,\sum^I_{i=1}\alpha^0_{i} =1\text{ and }\sum^J_{j=1}\beta^0_{j}=1.$

Let $I^0=\{1,\ldots,I\}\text{~~~~and~~~~}J^0=\{1,\ldots,J\}.$ Set $\Omega=\emptyset.$
\medskip

\noindent \begin{tabular}{l  l}
\hline \hline
\textbf{Algorithm 2:}& ~~~~~~~~~~~~~~~~~~~~~~~~~~~~~~~~~~~~~~~~~~~~~~~~~~~~~~~~~~~~~~~~~~~~~~~~~~~~~~~~~~~~~~~~~~~~~~~~~~~~~~~~~~~~~~~~~~~~~~~~~~~~~~~~~~~\\
\end{tabular}

\begin{minipage}{0.15\linewidth}

\noindent \textbf{Step $\boldsymbol{k\geq 1}$} \vspace{187pt}

\end{minipage}
\begin{minipage}{0.9\linewidth}

\texttt{IF} $|I^{k-1}|=1$ and $|J^{k-1}|=1$, 

\hspace{30pt} \texttt{THEN}, the procedure stops. 

\hspace{30pt} Set, $\gamma_k=\alpha_1^{k-1}=\beta_1^{k-1}$, $\tilde{\mu}_k^{x}=\mu_{I}^x$, $\tilde{\mu}_k^{y}=\mu_{J}^y.$ 

\hspace{30pt} Set $\Omega=\Omega\cup \{(\gamma_k, \tilde{\mu}_k^{x},\tilde{\mu}_k^{y})\}.$

\texttt{ELSE} ($|I^{k-1}|>1$ or $|J^{k-1}|>1$), the procedure continues as follows: 

\hspace{30pt} Set $\gamma_{k}=min\{\alpha^{k-1}_{1},\beta^{k-1}_{1}\}$.

\hspace{30 pt}\texttt{IF} $\gamma_k \neq \alpha_1 ^{k-1}$, 

\hspace{50 pt}\texttt{THEN}, set $~ I^k:=I^{k-1} $ and 
$
\alpha_{\ell}^{k}:=\left\{
\begin{array}{ll}
\alpha^{k-1}_{1}-\gamma_{k} & \text{~if~} \ell =1\\
\alpha^{k-1}_{\ell} & \text{ if }\ell>1\\
\end{array}
\right.
$,

\hspace{70 pt} for each $\ell \in I^{k-1}.$

\hspace{30 pt}\texttt{ELSE} ($\gamma_k = \alpha_1 ^{k-1}$), set $I^k:=I^{k-1}\setminus max_\ell \{\ell\in I^{k-1}\}$ and $\alpha_{\ell-1}^{k}= \alpha_{\ell}^{k-1}$ 

\hspace{70 pt} for each $\ell \in I^{k-1}.$
\end{minipage}
\begin{minipage}{0.15\linewidth}
\noindent  \vspace{20pt}

\end{minipage}
\begin{minipage}{0.9\linewidth}
\hspace{50 pt}\texttt{IF} $\gamma_k \neq \beta_1 ^{k-1}$,

\hspace{70 pt}\texttt{THEN}, $~ J^k :=J^{k-1} $ and
$
\beta_{\ell}^{k}:=\left\{
\begin{array}{ll}
\beta^{k-1}_{1}-\gamma_{k} & \text{~if~} \ell =1\\
\beta^{k-1}_{\ell} & \text{ if }\ell>1\\
\end{array}
\right.
$,

\hspace{90 pt} for each $\ell \in J^{k-1}.$

\hspace{50 pt}\texttt{ELSE}	($\gamma_k = \beta_1 ^{k-1}$), set $J^k:=J^{k-1}\setminus max_\ell \{\ell\in J^{k-1}\}$ and $\beta_{\ell-1}^{k}= \beta_{\ell}^{k-1}$ 

\hspace{90 pt} for each $\ell \in J^{k-1}.$

\hspace{50 pt} Set $p=|I^0|-|I^{k-1}|$ and $r=|J^0|-|J^{k-1}|$.

\hspace{50 pt} Set $\tilde{\mu}_k^{x}=\mu_{p+1}^x$ and $\tilde{\mu}_k^{y}=\mu_{r+1}^y$.

\hspace{50 pt} Set $\Omega=\Omega\cup \{(\gamma_k, \tilde{\mu}_k^{x},\tilde{\mu}_k^{y})\}$, and continue to Step k+1.
\bigskip

\end{minipage}
\noindent \begin{tabular}{l  l}
\hline \hline
 \textbf{}& ~~~~~~~~~~~~~~~~~~~~~~~~~~~~~~~~~~~~~~~~~~~~~~~~~~~~~~~~~~~~~~~~~~~~~~~~~~~~~~~~~~~~~~~~~~~~~~~~~~~~~~~~~~~~~~~~~~~~~~~~~~~~~~~~~~~~~~~~~~~~~~~~~~~~~~\\
\end{tabular}

\begin{lemma}
Algorithm 2 stops in a finite number of steps. That is, there is a $\tilde{k}$ such that $|I^{\tilde{k}-1}|=|J^{\tilde{k}-1}|=1$ and $\alpha_1^{\tilde{k}}=\beta_1^{\tilde{k}}.$ 
\end{lemma}
\begin{proof}
Note that in each step of Algorithm 2, we have that $|I^k|=|I^{k-1}|-1$ or $|J^k|=|J^{k-1}|-1.$ 
We also have that in each Step $k$ of the algorithm,
$$
\sum_{\ell\in I^{k}}\alpha_\ell^{k}=\sum_{\ell\in I^{k-1}}\alpha_\ell^{k-1}-\gamma_{k}~~~~\text{and}~~~\sum_{\ell\in J^{k}}\beta_\ell^{k}=\sum_{\ell\in J^{k-1}}\beta_\ell^{k-1}-\gamma_{k}.
$$
Hence, $$\sum_{\ell\in I^{k}}\alpha_\ell^{k}=\sum_{\ell\in I^{0}}\alpha_\ell^{0}-\sum_{t=1}^{k}\gamma_t=1-\sum_{t=1}^{k}\gamma_t. $$
Similarly, $$\sum_{\ell\in J^{k}}\beta_\ell^{k}=\sum_{\ell\in J^{0}}\beta_\ell^{0}-\sum_{t=1}^{k}\gamma_t=1-\sum_{t=1}^{k}\gamma_t. $$ 
That is, for each $k$ we have that
\begin{equation}\label{1 cuenta lema}
\sum_{\ell\in I^{k}}\alpha_\ell^{k}=\sum_{\ell\in J^{k}}\beta_\ell^{k}=1-\sum_{t=1}^{k}\gamma_t.
\end{equation}

By the finiteness of the sets $I^0$ and $J^0$, and given that in each step of Algorithm 2 we have that $|I^k|=|I^{k-1}|-1$ or $|J^k|=|J^{k-1}|-1$, we claim that  there is a $\tilde{k}$ such that $|I^{\tilde{k}-1}|=|J^{\tilde{k}-1}|=1.$
Assume that there is a Step $k_1-1$ such that $|I^{k_1-1}|=1\text{ and }|J^{k_1-1}|>1.$ By equality (\ref{1 cuenta lema}), we have that $\alpha_1^{k_1-1}= \sum_{\ell \in J^{k_1-1}}\beta_\ell^{k_1-1}.$ Hence. $\alpha_1^{k_1-1}> \beta_\ell^{k_1-1}$ for each $\ell \in J^{k_1-1}$ and $|I^{k_1}|=|I^{k_1-1}|$. Thus, $\alpha_{k_1}=\alpha_1^{k_1-1}-\gamma_{k_1}=\alpha_1^{k_1-1}-\beta_1^{k_1-1}$ and $J^{k_1}=J^{k_1-1}\setminus  max_\ell \{\ell\in J^{k_1-1}\}$, and $\beta_{\ell}^{k_1}= \beta_{\ell+1}^{k_1-1}$ for each $\ell \in J^{k_1}.$ Then, $|I^{k_1-1}|=|I^{k_1}|=1\text{ and }|J^{k_1}|=|J^{k_1-1}|-1 \geq 1.$ If $|J^{k_1}|>1$, then proceed with Algorithm 2 until there is a step $\tilde{k}$ such that $|I^{\tilde{k}-1}|=|J^{\tilde{k}-1}|=1$ and the procedure stops. Therefore, by equality (\ref{1 cuenta lema}), we have that $\alpha_1^{\tilde{k}}=\beta_1^{\tilde{k}}=\gamma_{\tilde{k}}.$
\end{proof}
\medskip

\noindent \begin{proof}[Proof of Proposition \ref{proposicion reescribir con algoritmo}]First we prove that there is $k_1$ such that $\alpha_1^{0}=\sum_{t=1}^{k_1}\gamma_t.$ Since $\gamma_1=min\{\alpha_1^0,\beta_1^0 \}$, we analyze two cases.
\begin{enumerate}[Case (i)]
\item[\textbf{Case 1:}] $\gamma_1=\alpha_1^0$. In this case $k_1=1$. 
\item[\textbf{Case 2:}] $\gamma_1<\alpha_1^0.$ In this case we have that $|I^0|=|I^1|$ and $\alpha_1^1=\alpha_1^0-\gamma_1.$ Then, in the next step, $\gamma_2 \leq \alpha_1^1.$

If $\gamma_2=\alpha_1^1$, then $\alpha_1^0=\gamma_1 + \gamma_2$.
 
If $\gamma_2<\alpha_1^1$, then repeat this procedure until $k_1$ is found such that $\gamma_{k_1}=\alpha_1^{k_1-1}$. Then $\alpha_1^0=\sum_{t=1}^{k_1}\gamma_t $. Note that $|I^0|=|I^1|=\ldots=|I^{k_1}|.$ Then, we have that $\tilde{\mu}^x_t=\mu^x_1$  for $t=1,\ldots,k_1$ and 
$$\sum_{t=1}^{k_1}\gamma_t \tilde{\mu}^{x}_t=\sum_{t=1}^{k_1}\gamma_t \mu^x_1 =\alpha_1^0 \mu^x_1.$$

Notice that $|I^{k_1}|=|I^{{k_1}-1}|-1$. That is, $1=p=|I^0|-|I^{k_1}|$ and $\tilde{\mu}_{k_1+1}^x=\mu_2^x.$ Then, for each $k\geq k_1+1$ we have that $\tilde{\mu}_k^x \neq \mu_1^x.$

Once we find $k_1$, we have to repeat this procedure with each  $\alpha_\ell^0$ for $\ell\geq2.$

\end{enumerate}

The case for $\beta$ is similar.
\end{proof}

We illustrate Algorithm 2 with two random matchings of Example 1.

\medskip

\noindent \textbf{Example 1 (Continued)} \textit{Let $x=\frac{1}{4}x^{\nu_1}+\frac{1}{2}x^{\nu_2}+\frac{1}{4}x^{\nu_4}$ and $y=\frac{1}{6}x^{\nu^1}+\frac{1}{2}x^{\nu^3}+\frac{1}{3}x^{\nu^4}$. Notice that both random stable matchings are represented as in Theorem \ref{teorema del orden}. We use Algorithm 2 to change their representation.
Let $I^0=\{1,2,3\}$ and $J^0=\{1,2,3\}$. Set $\Omega=\emptyset.$}
\begin{itemize}
\item[\textit{\textbf{Step 1}}]\textit{Since $I^0=\{1,2,3\}$ and $J^0=\{1,2,3\}$, set $\gamma_1=min\{\frac{1}{4},\frac{1}{6}\}=\frac{1}{6},$}

\begin{minipage}{0.5\linewidth}
\hspace{60pt}$
\begin{array}{|l}
\alpha_1^1=\frac{1}{4}-\frac{1}{6}=\frac{1}{12}  \\
\alpha_2^1=\frac{1}{2} \\
\alpha_3^1=\frac{1}{4} \\
\end{array}
$
\end{minipage}
\begin{minipage}{0.5\linewidth}
$
\begin{array}{|l}
 \beta_1^1=\frac{1}{2}\\
 \beta_2^1=\frac{1}{3}\\
\end{array}
$
\end{minipage}

\textit{Then, $I^1=\{1,2,3\}$, $J^1=\{1,2\}$, $\tilde{\mu}^x_1=\nu_1$ and $\tilde{\mu}^y_1=\nu_1$.
Set $\Omega=\Omega \cup \{(\nu_1,\nu_1,\frac{1}{6})\}$ and continue to Step 2.}

\item[\textit{\textbf{Step 2}}]\textit{Since $I^1=\{1,2,3\}$, $J^1=\{1,2\}$, set $\gamma_2=min\{\frac{1}{12},\frac{1}{2}\}=\frac{1}{12},$}

\begin{minipage}{0.5\linewidth}
\hspace{60pt}$
\begin{array}{|l}
\alpha_1^2=\frac{1}{2}\\
\alpha_2^2=\frac{1}{4} \\
\end{array}
$
\end{minipage}
\begin{minipage}{0.5\linewidth}
$
\begin{array}{|l}
 \beta_1^2=\frac{1}{2}-\frac{1}{12}=\frac{5}{12} \\
 \beta_2^2=\frac{1}{3}\\
 \end{array}
$
\end{minipage}

\textit{Then, $I^2=\{1,2\}$, $J^2=\{1,2\}$, $\tilde{\mu}^x_2=\nu_1$ and $\tilde{\mu}^y_2=\nu_3$. Set $\Omega=\Omega \cup \{(\nu_1,\nu_2,\frac{1}{12})\}$ and continue to Step 3.}

\item[\textit{\textbf{Step 3}}] \textit{Since $I^2=\{1,2\}$, $J^2=\{1,2\}$, set $\gamma_3=min\{\frac{1}{2},\frac{5}{12}\}=\frac{5}{12},$}

\begin{minipage}{0.5\linewidth}
\hspace{60pt}$
\begin{array}{|l}
\alpha_1^3=\frac{1}{2} -\frac{5}{12}=\frac{1}{12}\\
\alpha_2^3=\frac{1}{4} \\
\end{array}
$
\end{minipage}
\begin{minipage}{0.5\linewidth}
$
\begin{array}{|l}
 \beta_1^3=\frac{1}{4} \\
 \end{array}
$
\end{minipage}

\textit{Then, $I^3=\{1,2\}$, $J^3=\{1\}$, $\tilde{\mu}^x_3=\nu_2$ and $\tilde{\mu}^y_3=\nu_3$. Set $\Omega=\Omega \cup \{(\nu_2,\nu_3,\frac{5}{12})\}$ and continue to Step 4.}

\item[\textit{\textbf{Step 4}}]\textit{Since $I^3=\{1,2\}$, $J^3=\{1\}$, $\tilde{\mu}^x_3=\nu_2$, set $\gamma_4=min\{\frac{1}{12},\frac{1}{3}\}=\frac{1}{12},$ }

\begin{minipage}{0.5\linewidth}
\hspace{60pt}
$
\begin{array}{|l}
\alpha_1^4=\frac{1}{4} \\
\end{array}
$
\end{minipage}
\begin{minipage}{0.5\linewidth}
$
\begin{array}{|l}
 \beta_1^4=\frac{1}{3}-\frac{1}{12}=\frac{1}{4} \\
\end{array}
$
\end{minipage}

\textit{Then, $I^4=\{1\}$, $J^4=\{1\}$, $\tilde{\mu}^x_4=\nu_2$ and $\tilde{\mu}^y_4=\nu_4$. Set $\Omega=\Omega \cup \{(\nu_2,\nu_4,\frac{1}{12})\}$ and continue to Step 5.}

\item[\textit{\textbf{Step 5}}]\textit{Since $I^4=\{1\}$, $J^4=\{1\}$, then the procedure stops. Set $\gamma_6=min\{\frac{1}{4},\frac{1}{4}\}=\frac{1}{4},~\tilde{\mu}^x_6=\nu_4$ and $\tilde{\mu}^y_6=\nu_4$. Set $\Omega=\Omega \cup \{(\nu_4,\nu_4,\frac{1}{4})\}$ }
\end{itemize}
\textit{Therefore, we can represent the random stable matchings $x$ and $y$ as follows: }
\begin{center}
$
x=\frac{1}{6}x^{\nu_1}+\frac{1}{12}x^{\nu_1}+\frac{5}{12}x^{\nu_2}+\frac{1}{12}x^{\nu_2}+\frac{1}{4}x^{\nu_4},
$
\end{center}
\begin{center}
$y=\frac{1}{6}x^{\nu_1}+\frac{1}{12}x^{\nu_3}+\frac{5}{12}x^{\nu_3}+\frac{1}{12}x^{\nu_4}+\frac{1}{4}x^{\nu_4}.
$
\end{center}

Observe that $x$ and $y$ have five terms in each representation. Moreover, both lotteries have the same scalar, term to term.  
 \medskip
 
\begin{proof}[Proof of Proposition \ref{equivalencia de ordenes}]

\noindent $\boldsymbol{(\Longrightarrow)}$ Let $x$ and $y$ be two random stable matchings represented after the splitting procedure. Assume that $x\succeq_F y$. Fix $f\in F.$ We prove that $x\succeq_f^S y$. That is, $\tilde{\mu}^x_\ell (f)\geq_f \tilde{\mu}^{y}_\ell(f)$ for each $\ell=1,\ldots,\tilde{k}$.  

If $\tilde{\mu}_1^{y}(f)>_f \tilde{\mu}^{x}_1(f)$, we have that 
$$
0=\sum_{\ell:\tilde{\mu}^x_\ell (f)\geq_f \tilde{\mu}^{y}_1(f)} \gamma_\ell\geq \sum_{\ell:\tilde{\mu}^y_\ell (f)\geq_f\tilde{\mu}^{y}_1(f)} \gamma_\ell=\gamma_1>0,
$$
which is a contradiction. Then, $\tilde{\mu}_1^{x}(f)\geq_f \tilde{\mu}^{y}_1(f)$.
Assume that there is $k_1\leq \tilde{k}$ such that for each $\ell<k_1$ we have that $\tilde{\mu}_\ell^{x}(f) \geq_f\tilde{\mu}_\ell^{y}(f)$, and $\tilde{\mu}_{k_1}^{x}(f) <_f \tilde{\mu}_{k_1}^{y}(f)$.

Note that $\tilde{\mu}^{y}_\ell (f)\geq_f \tilde{\mu}_{\ell+1}^{y}(f)$ for each $\ell=1,\ldots,\tilde{k}-1$ implies that 

\begin{equation} \label{ecuacion para prop 2(1)}
\sum_{\ell=1}^{k_1}\gamma_\ell=\sum_{\ell:\tilde{\mu}^{y}_\ell (f)\geq_f \tilde{\mu}^{y}_{k_1}(f)} \gamma_\ell.
\end{equation}
By hypothesis ($x\succeq_F y$), in particular for $w=\tilde{\mu}_{k_1}^{y}(m)$ we have that 
$$
\sum_{\ell:\tilde{\mu}^y_\ell (f)\geq_f \tilde{\mu}^{y}_{k_1}(f)} \gamma_\ell \leq \sum_{\ell:\tilde{\mu}^x_\ell (f)\geq_f \tilde{\mu}^{y}_{k_1}(f)} \gamma_\ell.
$$
Notice that for $k_1,$ we have that $\tilde{\mu}_{k_1-1}^{x} (f)\geq_f \tilde{\mu}_{k_1-1}^{y}(f)$ and $\tilde{\mu}_{k_1}^{x}(f) <_f \tilde{\mu}_{k_1}^{y}(f)$. Then, 
$\tilde{\mu}_{k_1-1}^{x}(f)\geq_f \tilde{\mu}_{k_1-1}^{y}(f)\geq_f \tilde{\mu}_{k_1}^{y}(f) >_f \tilde{\mu}_{k_1}^{x}(f).$ 
Hence, 

\begin{equation}\label{ecuacion para prop 2(2)}
\sum_{\ell:\tilde{\mu}^x_\ell (f)\geq_f \tilde{\mu}^{y}_{k_1}(f)} \gamma_\ell = \sum_{\ell:\tilde{\mu}^x_\ell (f)\geq_f \tilde{\mu}^{x}_{k_1-1}(f)} \gamma_\ell=\sum_{\ell=1}^{k_1-1}\gamma_\ell.
\end{equation}
Thus, by equalities (\ref{ecuacion para prop 2(1)}) and (\ref{ecuacion para prop 2(2)}), we have that $\sum_{\ell=1}^{k_1}\gamma_\ell\leq \sum_{\ell=1}^{k_1-1}\gamma_\ell$, and this is a contradiction since $\gamma_{k_1}>0$. Then, there is no $k_1$ such that  for each $\ell<k_1$ we have that $\tilde{\mu}_\ell^{x}(f) \geq_f \tilde{\mu}_\ell^{y}(f)$, and $\tilde{\mu}_{k_1}^{x}(f) <_f \tilde{\mu}_{k_1}^{y}(f)$. Thus, $\tilde{\mu}_\ell^{x}(f) \geq_f \tilde{\mu}_\ell^{y}(f)$ for each $\ell=1,\ldots,\tilde{k},$  which in turns implies that $x\succeq_F^S y$.
\medskip
 
\noindent $\boldsymbol{(\Longleftarrow)}$ Recall that both $x$ and $y$ are represented by the splitting procedure. That is, both representations have the same numbers of terms and the same scalar term to term. Moreover, $\tilde{\mu}^{x}_\ell \geq_F \tilde{\mu}^{x}_{\ell+1}$ and $\tilde{\mu}^{y}_\ell \geq_F \tilde{\mu}^{y}_{\ell+1}$ for each $\ell=1,\ldots,\tilde{k}-1$. Also, since $x\succeq ^{S}_F y$, then $\tilde{\mu}^{x}_\ell \geq_F \tilde{\mu}^{y}_\ell$ for each $\ell=1,\ldots,\tilde{k}$.  Fix $\ell'$, then $$\{\gamma_\ell :\tilde{\mu}^{x}_\ell \geq_F \tilde{\mu}^{x}_{\ell'}\}=\{\gamma_\ell :\tilde{\mu}^{y}_\ell \geq_F \tilde{\mu}^{y}_{\ell'}\}\subseteq \{\gamma_\ell :\tilde{\mu}^{x}_\ell \geq_F \tilde{\mu}^{y}_{\ell'}\}.$$
Hence,   $$\sum_{\ell:\tilde{\mu}^{y}_\ell(f) \geq_f \tilde{\mu}^{y}_{\ell'}(f)} \gamma_\ell \leq \sum_{\ell:\tilde{\mu}^{x}_\ell (f)\geq_f \tilde{\mu}^{y}_{\ell'}(f)} \gamma_\ell$$
for each $f\in F$ and for each $\ell'=1,\ldots,\tilde{k}$.
Then, $x\succeq_F y$.

Therefore, the partial order $\succeq_F$ is equivalent to the domination relation $\succeq_F^S$.
\end{proof}
\medskip

\begin{proof}[Proof of Proposition \ref{teorema de operaciones binarias}]
We prove that $\boldsymbol{x\veebar_X y =\textbf{\textit{l.u.b.}}_{\succeq_X}(x,y)}$. Recall that by Proposition \ref{equivalencia de ordenes}, $x \succeq_F y$ if and only if $x\succeq_F^S y$ (analogously for $\succeq_W$ and $\succeq_W^S$).
\begin{enumerate}[(i)]

\item $\boldsymbol{x\veebar_X y \succeq_X x :}$

Since $\tilde{\mu}^{x}_{\ell}\vee_X \tilde{\mu}^{y}_{\ell} \geq_X \tilde{\mu}^{x}_{\ell}$ for each $\ell=1,\ldots,\tilde{k}$, then $x\veebar_X y \succeq^S_X x.$ Hence, $x\veebar_X y \succeq_X x.$

\item $\boldsymbol{x\veebar_X y \succeq_X y :}$

Since $\tilde{\mu}^{x}_{\ell}\vee_X \tilde{\mu}^{y}_{\ell} \geq_X \tilde{\mu}^{y}_{\ell}$ for each $\ell=1,\ldots,\tilde{k}$, then $x\veebar_X y \succeq^S_X y.$ Hence, $x\veebar_X y \succeq_X y.$

\item \textbf{If $\boldsymbol{z \succeq_X x}$  and  $\boldsymbol{z\succeq_X y}$, then $\boldsymbol{z \succeq_X x\veebar y}$:}

We have that $\tilde{\mu}^{z}_{\ell} \geq_X \tilde{\mu}^{x}_{\ell}$ and $\tilde{\mu}^{z}_{\ell} \geq_X \tilde{\mu}^{y}_{\ell}$ for each $\ell=1,\ldots,\tilde{k}$. Since, $\tilde{\mu}^{x}_{\ell}\vee_X \tilde{\mu}^{y}_{\ell}$ is the $\textit{l.u.b.}_{\geq_X}(\tilde{\mu}^{x}_{\ell}, \tilde{\mu}^{y}_{\ell}$), then $\tilde{\mu}^{z}_{\ell} \geq_X \tilde{\mu}^{x}_{\ell}\vee_X \tilde{\mu}^{y}_{\ell}$ for each $\ell=1,\ldots,\tilde{k}$. Hence, $z \succeq^S_X x\veebar_X y$. Therefore, $z \succeq_X x\veebar_X y$.
\end{enumerate}
The proof for $\boldsymbol{x\barwedge_X y =\textbf{\textit{g.l.b.}}_{\succeq_X}(x,y)}$ is analogous.

To prove that $\boldsymbol{x\veebar_F y =x\barwedge_W y},$ recall that the lattices of stable matchings are dual, that is, given $\mu_1,~\mu_2 \in \mathcal{S(P)}$  $\mu_1\vee_F\mu_2= \mu_1\wedge_W \mu_2$. 
By Definition \ref{defino operaciones binarias}, we have that if $0< \gamma_{\ell} \leq 1$, $\sum_{\ell=1}^{\tilde{k}}\gamma_{\ell}=1$, $\tilde{\mu}^{x}_{\ell}\in{\mathcal{S(P)}}$ , $\tilde{\mu}^{x}_{\ell} \geq_{F} \tilde{\mu}^{x}_{\ell+1}$ and $\tilde{\mu}^{y}_{\ell} \geq_{F} \tilde{\mu}^{y}_{\ell+1}$, then
$$x\veebar_F y= \sum_{\ell=1}^{\tilde{k}}\gamma_{\ell}(\tilde{\mu}^{x}_{\ell}\vee_F \tilde{\mu}^{y}_{\ell})=\sum_{\ell=1}^{\tilde{k}}\gamma_{\ell}(\tilde{\mu}^{x}_{\ell}\wedge_W \tilde{\mu}^{y}_{\ell})=x\barwedge_W y.$$
 The proof for $\boldsymbol{x\veebar_W y =x\barwedge_F y}$  is analogous.
\end{proof}
\end{document}